\documentclass{lmcs} 
\pdfoutput=1

\usepackage{lastpage}
\lmcsdoi{16}{3}{16}
\lmcsheading{}{\pageref{LastPage}}{}{}%
{Feb.~23,~2019}{Sep.~08,~2020}{}

\keywords{two-variable logic, finite model theory, algebraic automata theory}

\usepackage[utf8]{inputenc}
\usepackage{amssymb}
\usepackage{graphicx}
\usepackage{stmaryrd}
\usepackage{tikz}
\usetikzlibrary{arrows,automata}
\usepackage{xcolor,booktabs,colortbl,tabularx}

\newcommand{\cJ}{\mathcal{J}}

\newcommand{\suc}{\texttt{suc}}
\newcommand{\betw}{\texttt{bet}}
\newcommand{\betfac}{\texttt{betfac}}
\newcommand{\thr}{\texttt{th}}
\newcommand{\thrfac}{\texttt{thfac}}
\newcommand{\thtl}{\mbox{$\mathit{ThTL}$}\xspace}
\newcommand{\bthtl}{\mbox{$\mathit{BThTL}$}\xspace}
\newcommand{\invtl}{\mbox{$\mathit{InvTL}$}\xspace}
\newcommand{\binvtl}{\mbox{$\mathit{BInvTL}$}\xspace}
\newcommand{\factl}{\mbox{$\mathit{NFacTL}$}\xspace}
\newcommand{\bfactl}{\mbox{$\mathit{BFacTL}$}\xspace}
\newcommand{\bthfactl}{\mbox{$\mathit{BThFacTL}$}\xspace}
\newcommand{\utl}{\mbox{$\mathit{TL}[\fut,\past]$}}
\newcommand{\until}{\textsf{U}}
\newcommand{\since}{\textsf{S}}
\newcommand{\fut}{\textsf{F}}
\newcommand{\henceforth}{\textsf{G}}
\newcommand{\past}{\textsf{P}}
\newcommand{\nextt}{\textsf{X}}
\newcommand{\prev}{\textsf{Y}}
\newcommand{\gfut}[1]{\mbox{$\fut_{#1}$}}
\newcommand{\gpast}[1]{\mbox{$\past_{#1}$}}

\newcommand{\ltl}{\mbox{$\mathit{LTL}$}\xspace}
\newcommand{\ltlbin}{\mbox{$\ltl[\until,\since]$}\xspace}
\newcommand{\cltlbin}{\mbox{$C\ltlbin$}\xspace}
\newcommand{\ltlun}{\mbox{$\ltl[\fut,\past]$}\xspace}
\newcommand{\ltlunsuc}{\mbox{$\ltl[\fut,\past,\nextt,\prev]$}\xspace}
\newcommand{\ltlunsucp}{\mbox{$\ltl[\fut,\past,\nextt^{n},\prev^{n}]$/}\xspace}

\usepackage{xspace}

\newcommand{\defn}{\mathrel{\mbox{$~\stackrel{\rm def}{=}~$}}}
\renewcommand{\epsilon}{\varepsilon}
\newcommand{\limplies}{\rightarrow}
\newcommand{\eqvt}{\leftrightarrow}
\newcommand{\true}{true}
\newcommand{\false}{false}
\newcommand{\orover}{{\displaystyle\bigvee}}
\newcommand{\andover}{{\displaystyle\bigwedge}}
\newcommand{\diam}[1]{\langle#1\rangle}
\newcommand{\boxm}[1]{[#1]}
\newcommand{\exact}[1]{\llbracket#1\rrbracket}
\newcommand{\expand}[1]{\{#1\}}
\newcommand{\fo}{\mbox{$\mathit{FO}$}\xspace}
\newcommand{\foless}{\mbox{$\fo[<]$}\xspace}
\newcommand{\fotwo}{\mbox{$\fo^2$}\xspace}
\newcommand{\fotwoless}{\mbox{$\fotwo[<]$}\xspace}
\newcommand{\fotwosuc}{\mbox{$\fotwo[<,\suc]$}\xspace}
\newcommand{\fotwobet}{\mbox{$\fotwo[<,\betw]$}\xspace}
\newcommand{\fotwobetfac}{\mbox{$\fotwo[<,\betfac]$}\xspace}
\newcommand{\fotwothr}{\mbox{$\fotwo[<,\thr]$}\xspace}
\newcommand{\fotwothrfac}{\mbox{$\fotwo[<,\thrfac]$}\xspace}
\newcommand{\pspace}{\mbox{$\mathit{Pspace}$}}
\newcommand{\expspace}{\mbox{$\mathit{Expspace}$}}
\newcommand{\nexptime}{\mbox{$\mathit{NExptime}$}}
\newcommand{\np}{\mbox{$\mathit{NP}$}}

\newcommand{\fosucc}{FO^2[<,{\tt succ}]}
\newcommand{\fobet}{FO^2[<,\betw]}
\newcommand{\fobetfac}{FO^2[<,\betfac]}

\newcommand{\da}{{\bf DA}\xspace}
\newcommand{\me}{{\bf M}_e}
\newcommand{\meda}{\me{\bf DA}}

\newcommand{\stard}{*{\bf D}}

\newcommand{\medad}{\meda\stard}

\begin{document}

\title[Two-variable logics with some betweenness relations]{Two-variable logics with some betweenness relations: Expressiveness, satisfiability and membership}

\author[A.~Krebs]{Andreas Krebs\rsuper{a}}	
\address{\lsuper{a}Universit\"at T\"ubingen}	

\author[K.~Lodaya]{Kamal Lodaya\rsuper{b}}	
\address{\lsuper{b}The Institute of Mathematical Sciences, Chennai, and
	Homi Bhabha National Institute, Mumbai}

\author[P.K.~Pandya]{Paritosh K.~Pandya\rsuper{c}}	
\address{\lsuper{c}Tata Institute of Fundamental Research, Mumbai}	

\author[H.~Straubing]{Howard Straubing\rsuper{d}}
\address{\lsuper{d}Boston College}

\begin{abstract}
We study two extensions of  \fotwoless, first-order logic interpreted in finite words, in which formulas are restricted to use only two variables.  We adjoin to this language two-variable atomic formulas that say, `the letter $a$ appears between positions $x$ and $y$' and `the factor $u$ appears between positions $x$ and $y$'.  These are, in a sense, the simplest properties that are not expressible using only two variables.

We present several logics, both first-order and temporal, that have the same expressive power, and find matching lower and upper bounds for the complexity of satisfiability for each of these formulations.
We give effective conditions, in terms of the syntactic monoid of a regular language, for a property to be expressible in these logics. This algebraic analysis allows us to prove, among other things, that our new logics have strictly less expressive power than full first-order logic \foless.
Our proofs required the development of novel
techniques concerning factorizations of words.
\end{abstract}

\maketitle

\nocite{L20}
\section{Introduction}

We denote by $FO[<]$ first-order logic with the order relation $<$, interpreted in finite words over a finite alphabet $A$.  Variables in first-order formulas are interpreted as positions in a word, and for each letter  $a\in A$ there is a unary relation $a(x)$, interpreted to mean `the letter in position $x$ is $a$'.
We also use the binary successor relation $\suc(x,y)$ to denote that the position
$y$ immediately follows the position $x$, that is, $y=x+1$.
Thus sentences in this logic define properties of words, or, what is the same thing, languages $L\subseteq A^*$. The logic $FO[<]$ over words has been extensively studied, and has many equivalent characterizations in terms of temporal logic, regular languages, and the algebra of finite semigroups.
See, for instance,~\cite{Str-book,Wilke} and the many references cited therein.

The first-order definable languages---those definable
in the logic \foless---were shown equivalent to languages definable by star-free expressions
by the work of McNaughton and Papert~\cite{MNP} and Sch\"utzenberger~\cite{Sch1}.
The algebraic viewpoint established decidability of the definability question,
that is, whether a given regular language is first-order definable.
Much subsequent interest focused on effectively determining the quantifier alternation depth of definable languages.
The work of Simon~\cite{Simon} on `piecewise-testable events' provides such a characterization for languages
definable by boolean combinations of $\Sigma_1$ sentences.
Recently Place and Zeitoun found characterizations for several higher levels of the hierarchy~\cite{PZ-higher}.

If we allow variable symbols to be reused, then even formulas with high quantifier depth can be rewritten as equivalent formulas
with a small number of distinct variable symbols: Kamp proved~\cite{Kamp} that every sentence of $FO[<]$ is equivalent to one using only three variables. The family of languages definable with two-variable sentences is strictly smaller (see, for example, Immerman and Kozen~\cite{IK}). The fragment $FO^2[<]$, consisting of the two-variable formulas, has also been very thoroughly investigated, and once again, there are many equivalent effective characterizations~\cite{TW}.

The reason $FO^2[<]$ is strictly contained in $FO[<]$ is that one cannot express `betweenness' with only two variables.  More precisely, let $a\in A$.  Then the predicate
\[a(x,y)=\exists z(x<z<y \wedge a(z)),\]
which says that there is an occurrence of $a$ strictly between $x$ and $y$, is not expressible using only two variables.  We denote by \fotwobet the two-variable first-order logic that results from adjoining these predicates for each $a\in A$ to $FO^2[<]$.

More generally, if $u=a_1\cdots a_n\in A^+$, then we can consider the predicate $\diam{u}(x,y)$, which says that there is an occurrence of the factor $u$ strictly between $x$ and $y$. We denote by  $\fotwobetfac$ the two-variable logic that results when we adjoin all such predicates to  $FO^2[<]$.


What properties of words can we express when we adjoin these new relations to $FO^2[<]$?  The first obvious question to ask is whether we recover all of $FO[<]$ in this way.  The answer, as we shall see, is `no', but we will give a much more precise description.

The present article is a study of these extended two-variable logics \fotwobet and \fotwobetfac.  Our investigation is centered around two quite different themes.  One theme investigates several different
logics, based on $FO[<]$ as well as on the temporal logic $LTL$,  for expressing each of these betweenness, and establishes their expressive equivalence. We explore the complexity of
satisfiability-checking in these logics as a measure of their descriptive succinctness.

The second theme is devoted to determining, in a sense that we will make precise, the exact expressive power of the logics \fotwobet  and \fotwobetfac.  Here we draw on tools from the algebraic theory of semigroups to find decision procedures for determining whether a given regular language is definable in \fotwobet  or in \fotwobetfac.

In Section~\ref{sec:BasicProperties} we  give the precise definition of our logics
(although there is not much more to it than what we have written in this Introduction).
We introduce related logics  \fotwothr and  \fotwothrfac which enforce quantitative constraints on counts of letters and factors, respectively. We
show that they have the same expressive power as  \fotwobet and \fotwobetfac, respectively, although they can result in formulas that are considerably more succinct.
Quite how far this expressiveness goes in terms of
the quantifier alternation depth (or, what is more or less the same thing, the so-called `dot-depth') of languages definable in this logic
is studied in Section~\ref{sec:altdepth}.

In Section~\ref{sec:tlsat} we introduce temporal logics,
qualitative  and quantitative, but again with the same expressive power as our original formulations.
We determine the complexity of formula
satisfiability for all these temporal as well as two-variable logics.

Section~\ref{AndreasHoward} introduces the algebraic machinery we will need
to characterize the expressiveness of these logics.

Sections~\ref{sec:licsmain} and~\ref{sec:fac} are devoted to a characterization of the expressive power of $\fobet$ in terms of the algebra of finite semigroups, and
Section~\ref{sec:cslnext} reduces the definability question for  $\fobetfac$ to that of $\fobet$. As a result we find effective conditions both  for a language to be definable \fotwobet and in $\fobetfac$.
 We use these results to show that \fotwobetfac is strictly less expressive than $FO[<]$, and $\fotwobet$ is strictly less expressive that $\fotwobetfac$. Moreover, we also show that
 $\fotwoless$ is strictly less expressive than $\fotwobet$.

\begin{table*}[t]{\small
\setlength{\aboverulesep}{0pt}
\setlength{\belowrulesep}{0pt}
\setlength{\extrarowheight}{1ex}
\begin{center}
  \begin{tabularx}{\textwidth}{>{\raggedright}l>{\raggedright}l>{\raggedright\columncolor{gray!50}}l>{\raggedright\columncolor{gray!50}}>{\raggedright}l>{\raggedright}l>{\raggedright\arraybackslash}l}
    \toprule
    Complexity/ &      &                     &                 &              & \\
    Variety & {\bf Ap} & {\bf M${}_e$DA$*$D} & {\bf M${}_e$DA} & {\bf DA$*$D} & {\bf DA} \\[1ex]
    \midrule[0.08em]
    Nonelementary      & $\fo[<]$ &  &  &  &  \\[1ex]
    \midrule
    \expspace          & \cltlbin & \fotwobetfac, & \fotwothr, & \ltlunsucp & \\
                       &          & \fotwothrfac, & \fotwobet, & (binary~notation) & \\
                       &          & \bthfactl     & \bthtl, &                & \\
                       &          &           & \thtl &                & \\[1ex]
    \midrule
    \nexptime          &          &           &       & \fotwosuc      & \fotwoless \\
                       &          &           &       &                & (unbounded \\
                       &          &           &       &                & alphabet) \\[1ex]
    \midrule
    \pspace            & \ltlbin  & \bfactl,  & \binvtl, & \ltlunsuc   & \\
                       &          & \factl    & \invtl &             & \\[1ex]
    \midrule
    \np                &          &           &       &              & \fotwoless \\
                       &          &           &       &              & (bounded \\
                       &          &           &       &              & alphabet), \\
                       &          &           &       &              & \ltlun \\[1ex]
    \bottomrule
  \end{tabularx}
\end{center}
\caption{A summary of the results in this paper, in the context of the complexity of satisfiability-checking and the expressive power of temporal and predicate logics studied in earlier work. The top row shows the varieties of finite semigroups and monoids that provide the effective characterization
of the corresponding logics.  Our new results are in the shaded columns.}\label{tab:results}
}\end{table*}


Table~\ref{tab:results} gives a map of our results
and compares them to those of related previous work. Etessami {\em et al.}~\cite{EVW} as well as Weis and Immerman~\cite{WI} have explored logics \fotwoless and \fotwosuc, as well as matching temporal logics and their decision complexities.
Our own work has been in interval logics with the same expressive power~\cite{LPS1,LPS}.
Th\'erien and Wilke~\cite{TW} found characterizations of the expressive power of these same logics, using algebraic methods. We find that our new logics are more expressive but this comes at the cost of some computational power.

In terms of related work,
some counting extensions \cltlbin of \ltlbin have been studied by
Laroussinie {\em et al.}~\cite{LMP}, and by Alur and Henzinger as discrete time Metric Temporal logic~\cite{AH}.
In a more general setting,
satisfiability and model theory of two-variable and guarded logics with several relations on ordered as well as unordered relational structures have been intensively studied (see Otto~\cite{otto01}, Gr\"adel~\cite{gra99}).


\section{Two-variable logics and games}\label{sec:BasicProperties}

\subsection{Two-variable Between Logic}%
\label{sec:examples}
Throughout this paper, $A$ denotes a finite alphabet, $A^*$ the set of all words over $A$ (including the empty word), and $A^+$ the set of all nonempty words over $A$.  If $w\in A^+$, then $|w|$ denotes the length of $w$, and if, further, $1\leq i\leq |w|$, then $w(i)$ denotes the $i^{th}$ letter of $w$, where we take the leftmost letter of $w$ to be the first letter.  The \textit{content} $\alpha(w)\subseteq A$ of a word $w\in A^*$ is the set of letters that it contains.

$FO[<]$ is first-order logic interpreted in words over a finite alphabet $A$. Variables represent positions in a word, and the binary relation $<$ is interpreted as the usual ordering on positions.  There is also   a unary relation $a(x)$ for each $a \in A$, interpreted to mean that the letter in position $x$ is $a$. If $\phi$ is a sentence of $FO[<]$, then the set of words $w\in A^*$ such that $w\models\phi$ is a language in $A^*$, in fact a regular language. Similarly, a formula $\phi(x)$  with a single free variable defines a set of \emph{marked words} $(w,i)$, where $w\in A^+$ and $1\leq i\leq |w|$.
Often we will be a little sloppy in our terminology, and treat $FO[<]$ and its various sublogics at times as sets of formulas, or as sets of sentences, or as a family of languages, or as a family of sets of marked words.\footnote{A marked word cannot be empty, but a word can be. In contrast to the usual practice in model theory, we permit our formulas to be interpreted in the empty word:  every existentially quantified sentence is taken to be false in the empty word, and thus every universally quantified sentence is true.}

 For each $a\in A$ we adjoin to this logic a \textit{binary} relation $a(x,y)$ which is interpreted to mean
$\exists z(x<z\wedge z<y\wedge a(z))$.
This relation cannot be defined in ordinary first-order logic over $<$ without introducing a third variable.  We will investigate the fragment \fotwobet, obtained by restricting to formulas that use  both the unary and binary $a$ relations, along with $<$, but use only two variables.

\begin{exa}
There is an even simpler relation that is not expressible in two-variable logic that we could have adjoined:  this is the successor $y=x+1$ (which we also write $\suc(x,y)$). The logic \fotwoless supplemented by successor, which we denote by \fotwosuc has also been extensively studied, and the kinds of questions that we take up here for \fotwobet have already been answered for \fotwosuc. (See, for example,~\cite{EVW,LPS,TW}).
\end{exa}

\begin{exa}
The successor relation  $y=x+1$ is itself definable in \fotwobet, by a formula that says no letter of the alphabet appears strictly between $x$ and $y$.
As a result, we can define the set $L$  of words over $\{a,b\}$ in which there is no occurrence of two consecutive $b$'s by a sentence of \fotwobet. We can similarly define the set of words without two consecutive $a$'s.  Since we can also say that the first letter of a word is $a$ (by $\forall x(\forall y (x\leq y)\rightarrow a(x))$ and that the last letter is $b$, we can define the language ${(ab)}^*$ in \fotwobet).  This language is not, however, definable in \fotwo$[<]$.
\end{exa}

\begin{exa}
Let $L\subseteq {\{a,b\}}^*$ be the language defined by the regular expression
\[{(a+b)}^*bab^+ab{(a+b)}^*.\]
This language is definable in \fotwobet\ by the sentence
\[\exists x(\exists y(b(x,y)\wedge\neg a(x,y)\wedge \phi(y))\wedge\psi(x)),\]
where
\[\phi(y)=a(y)\land \exists x(\suc(y,x)\wedge b(x))\text{, and}\]
\[\psi(x)=a(x)\land \exists y(\suc(y,x)\wedge b(y)).\]
As we shall see further on, this language is not definable in \fotwosuc, so our new logic has strictly more expressive power than \fotwosuc.
\end{exa}

\begin{exa}
The value of an $r$-bit counter (modulo $2^r$) can be represented by a word $b_1 \cdots b_{r}$ over  $\{0,1\}$,  with $b_1$ representing the least significant bit.  A sequence $c_1 c_2 \ldots c_k$ of  $k$ such values can be represented by a word
$mark \cdot b^1_1 \cdots b^1_r \cdot mark \cdot b^2_1 \cdots b^2_r \cdot mark \cdots mark \cdot b^k_1 \cdots b^k_r$ where $b^i_j \in \{0,1\}$ and $mark$ is a new letter used to separate two successive counter values.
The logic \fotwobet (and in fact its sublogic \fotwosuc) can  assert
properties such as $(c_j=c_i)$, $(c_i=p)$, $(c_j=c_i+1)$ or even $(c_i = c_j + p)$ for any constant $p$ with $0 \leq p < 2^r$ using formulas of
length polynomial in $r$. These formulas use monadic predicates $mark(x), 0(x), 1(x)$.

The formulas $Bit^i_0(x)$ for $1\leq i \leq r$ when evaluated at position $x$ having letter $mark$ states that the bit in position $x+i$ has value $0$.
This is defined as
\[
\begin{array}{rl}
Bit^1_0(x)   \defn & \!\!\!\! \exists y: \suc(x,y) \land 0(y),\\
Bit^{i+1}_0(x) \defn & \!\!\!\! \exists y: \suc(x,y) \land Bit^{i}_0(y),
\end{array}
\]
We define the formulas $Bit^i_1(x)$ similarly.

The $O(r^2)$ size formula $EQ(x,y)$ below checks equality of two numbers
by comparing the $r$ bits in succession.
We use the fact that the bit string always has $r$ bits.

\[\begin{array}{rl}
EQ(x,y)   \defn & mark(x) \land mark(y) \land \andover_{i=1}^{r} EQ_i(x,y),\\
EQ_i(x,y) \defn & (Bit^i_0(x) \eqvt Bit^i_0(y))
\end{array}\]

By small variations of this formula, we can define formulas
$LT$, $GT$ etc, to make other comparisons.
Incrementing the counter modulo $2^r$
can be encoded by an $O(r^3)$ formula $INC_1(x,y)$
which converts a least significant block of $1$s to $0$s,
and $0$ after that to $1$.
We can also define $INC_c(x,y)$ which checks that the number at position $y$
of the word is obtained by increasing the number at position $x$
by a constant $c$.
\end{exa}

In contrast, it is quite difficult to find examples of languages definable in $FO[<]$ that are \textit{not} definable in \fotwobet.  Much of this paper is devoted to establishing methods for generating such examples.

More generally, for each $u=a_1\cdots a_n\in A^+$ we define a binary relation $\diam{u}(x,y)$, which is interpreted to mean
\[\exists z_1\dots \exists z_n(
x < z_1 < \dots z_n < y \land \suc(z_1,z_2) \land \dots \land \suc(z_{n-1},z_n)
\land a_1(z_1) \land \dots \land a_n(z_n)).\]
That is, there is an occurrence of the factor $u$ between $x$ and $y$.  The two-variable logic obtained by adjoining all these relations to $FO^2[<]$ is denoted $\fobetfac$. In contrast to  \fotwobet,  we will only interpret sentences in this logic in nonempty words. (This somewhat technical condition is required to make our algebraic characterizations work. We will have more to say about this in Section~\ref{AndreasHoward}.)

\subsection{Two-variable Threshold Logic}%
\label{sec:twotwo}

We  generalize \fotwobet as follows: Let $k\geq 0$ and $a\in A$.  We define $(a,k)(x,y)$  to mean that $x<y$, and that there are at least $k$ occurrences of $a$ between $x$ and $y$.   Adding these (infinitely many) relations gives a new logic $\fotwothr$.  In an analogous fashion, we define \fotwothrfac by adjoining relations asserting that there are at least $k$ (possibly overlapping) occurrences of a factor $u\in A^+$ strictly between $x$ and $y$.

\begin{exas}
The language $STAIR_k$ consists of all words $w$ over $\{a,b,c\}$ which have a factor of the form $a {(a+c)}^* a $ with at least $k$
occurrences of $a$. This can be specified by sentence $\exists x \exists y ( (a,k)(x,y) \land \neg b(x,y))$.

Threshold logic is quite useful in specifying quantitative properties of systems. For example, a bus arbiter circuit may have the property that if $req$ is continuously on for $15$ cycles then there should be at least $3$ occurrences of $ack$. This can be specified by
$\forall x \forall y ((req,15)(x,y) \rightarrow (ack,3)(x,y))$.
\end{exas}

Since $a(x,y)$ is equivalent to $(a,1)(x,y)$,  $\fotwothr$ is at least as expressive as \fotwobet.  What is less obvious is that the converse is true, albeit at the cost of a large blowup in the quantifier complexity of formulas:

\begin{thm}\label{thm.invequalsthr}
Considered both as families of languages and as families of sets of marked words,
\[\fotwobet=\fotwothr.\]
\end{thm}

There is a bit more to this than meets the eye---the stated equality of expressive power holds only for sentences and for formulas with a single free variable interpreted in finite words or finite marked words.
For instance, the relations $(a,k)(x,y)$ for $k>1$ are not themselves expressible by single formulas of \fotwobet, and therefore the proof of Theorem~\ref{thm.invequalsthr} is not completely straightforward. We will give the proof of Theorem~\ref{thm.invequalsthr} in the next subsection, after introducing our game apparatus below.

We also have the analogous result for the other logics we have introduced here:

\begin{thm}\label{thm.invequalsthrfac}
Considered both as families of languages and as families of sets of marked words,
\[\fotwobetfac=\fotwothrfac.\]
\end{thm}

We will give the proof of this in Section~\ref{sec:cslnext}.

\medskip

We will prove Theorem 2.1 using a game-based argument.  We   define two games, one characterizing expressibility in  $FO^2[<,bet]$, and the other expressibility in $FO^2[<,th]$.  These are variants of the standard Ehrenfeucht-Fra\"{\i}ss\'e games.  We then argue that the existence of a winning strategy for the second player  in either one of the games implies the existence of a winning strategy in the other, although with a different number of rounds.

\subsection{A game characterization of \texorpdfstring{$\pmb{\fotwobet}$}{FO2[<,BET]}}\label{sec:efgame}

We write $(w_1,i_1)\equiv_k (w_2,i_2)$ if these two marked words satisfy exactly the same formulas of \fotwobet\  with one free variable of quantifier depth no more than $k$.

We overload this notation, and also write $w_1\equiv_k w_2$ if $w_1$ and $w_2$ are ordinary words that satisfy exactly the same \textit{sentences} of \fotwobet with quantifier depth $\leq k$.

Let $k\geq 0$. The game is played for $k$ rounds in two marked words $(w_1,i_1)$ and $ (w_2,i_2)$ with a single pebble on each word.  At the start of the game, the pebbles are on the marks $i_1$ and $i_2$.   In each round, the pebble is moved to a new position in both words, producing two new marked words.

Suppose that at the beginning of a round, the marked words are $(w_1,j_1)$ and $(w_2,j_2)$.  Player 1 selects one of the two words and moves the pebble to a different position.  Let's say he picks $w_1$, and moves the pebble to $j_1'$, with $j_1\neq j_1'$.  Player 2 moves the pebble  to a new position $j_2'$ in $w_2$.  This response is required to satisfy the following properties:
\begin{enumerate}
\renewcommand\labelenumi{(\roman{enumi})}
\item The moves are in the same direction:  $j_1<j_1'$ iff $j_2<j_2'$.
\item The letters in the destination positions are the same: $w_1(j_1')=w_2(j_2')$.
\item The set of letters jumped over is the same---that is, assuming $j_1<j_1'$:
\begin{eqnarray*}
&\hspace{-5mm}&\{a\in A: w_1(k)=a\text{ for some } j_1<k<j_1'\}=\\
&\hspace{-5mm}&\{a\in A: w_2(k)=a\text{ for some } j_2<k<j_2'\}.
\end{eqnarray*}
\end{enumerate}

\noindent
Player 2 wins the 0-round game if $w_1(i_1)=w_2(i_2)$. Otherwise, Player 1 wins the 0-round game.

Player 2 wins the $k$-round game for $k>0$ if she makes a legal response in each of $k$ successive rounds, otherwise Player 1 wins.

The following theorem and its corollary are just the standard results about Ehrenfeucht-Fra\"{\i}ss\'e games adapted to this logic; we omit the proofs.

\begin{thm}\label{thm.game_equiv}
$(w_1,i_1)\equiv_k(w_2,i_2)$ if and only if Player 2 has a winning strategy in the $k$-round game in the two marked words.
\end{thm}

We can also define the $k$-round game in ordinary unmarked words $w_1,w_2\in A^*$.  Player 1 begins in the first round by placing a pebble on a position in one of the two words, and Player 2 must respond on a position in the other word containing the same letter.  Thereafter, they play the game in the two marked words that result for $k-1$ rounds.  The following is a direct consequence of the preceding theorem.

\begin{cor}\label{cor.game_equiv}  Player 2 has a winning strategy in the $k$-round game in $w_1$ and $w_2$ if and only if $w_1\equiv_k w_2$.
\end{cor}


\subsection{Proof of Theorem~\ref{thm.invequalsthr}}

We introduce a game characterizing $\fotwothr$. Let $\theta$ be a function from $A$ to the positive integers. We consider formulas in $\fotwothr$ in which for all $a\in A$, every occurrence of the relation $(a,k)(x,y)$ has $k\leq\theta(a)$.  Let's call these \textit{$\theta$-bounded} formulas.

The rules of the game are the same as those for the \fotwobet game, with this difference:  At each move, for each $a\in A$, the number $m_1$ of $a$'s jumped by Player 1 must be equivalent, threshold $\theta(a)$, to the number $m_2$ of $a$'s jumped by Player 2. That is, either $m_1$ and $m_2$ are both greater than or equal to $\theta(a)$, or $m_1=m_2$.  Observe that the game for \fotwobet is the case $\theta(a)=1$ for all $a\in A$.

For marked words $(w_1,i_1)$ and $(w_2,i_2)$, let us define, $(w_1,i_1)\equiv^{\theta}_k(w_2,i_2)$ if and only if they satisfy exactly the same $\theta$-bounded formulas of quantifier depth less than or equal to $k$.  As with the case of \fotwobet, we also have a version of both the game and the equivalence relation for ordinary words.

It is a routine matter to show that the analogues of Theorem~\ref{thm.game_equiv} and Corollary~\ref{cor.game_equiv} hold in this more general setting: Player 2 has a winning strategy in the $k$-round game in $(w_1,i_1),(w_2,i_2)$ if and only if $(w_1,i_1)\equiv^{\theta}_k(w_2,i_2)$, and likewise for ordinary words.

Let $\theta, \theta'$ be two functions from $A$ to the natural numbers that differ by one in the following sense:  $\theta'(b)=\theta(b)+1$ for exactly one $b\in A$, and $\theta'(a)=\theta(a)$ for all $a\neq b$.  Obviously $\equiv^{\theta'}_k$ is finer than $\equiv^{\theta}_k$.  We claim that  $\equiv^{\theta}_{2k}$ refines $\equiv^{\theta'}_{k}$.

This will give  the desired result, because any threshold function $\theta$ can be built from the base threshold function that assigns 1 to each letter of the alphabet by a sequence of steps in which we add 1 to the threshold of each letter.  So it follows by induction that for any $\theta$, $\equiv^{\theta}_k$ is refined by $\equiv_{k\cdot 2^r}$, where $r=\sum_{a\in A}\theta(a)-1$.  In particular, each $\equiv^{\theta}_k$-class is definable by a \fotwobet sentence, although the quantifier depth of this sentence is exponential in the thresholds used.

So given a fixed $\fotwothr$ formula $\phi$ there is a threshold $\theta$ such that the formula is $\theta$-bounded. Let $k$ be the quantifier depth of $\phi$. Then the formula cannot distinguish words that are in the same equivalence class with respect to $\equiv^{\theta}_k$. As there are only finitely many $\theta$-bounded formulas of quantifier depth at most $k$, there are only finitely many such equivalence classes. By the argument above we can find a \fotwobet\ formula for each equivalence class accepted by $\phi$ and the disjunction of these will be a formula in \fotwobet\ that has the same models as $\phi$.

Finally, we prove the claim above with a game argument, showing that if Player 2 has a winning strategy in the $2k$-round $\theta$ game in $(w_1,w_2)$ then she has a strategy in the $k$-round $\theta'$ game in the same two words.   This is done by a strategy-copying argument:  Suppose Player 2 needs to know how to reply to a rightward move by Player 1 in $w_1$.  If for each $a$, the number of occurrences of $a$ in the jumped-over letters is no more than $\theta(a)$, then she can just use the reply she would ordinarily make in the $\theta$ game.  Suppose, however, that the jumped-over letters contain $\theta'(b)=\theta(b)+1$ occurrences of $b$.  Player 2 computes her reply by decomposing this move of Player 1 into two separate moves:  A jump to an occurrence of $b$ strictly between the source and destination of the move, and then a jump from this position to the destination of the move.  Player 2 has a legal reply to each of these moves in the $\theta$-game, and the two successive replies constitute  a successful reply in the $\theta'$-game.  Observe that Player 2 uses up no more than two of her moves in the $\theta$-game for each move in the $\theta'$ game.


\section{Alternation depth}\label{sec:altdepth}
A first-order formula that begins with a sequence of quantifiers, followed by a quantifier-free formula, is said to be in \textit{prefix form}.  The sequence of quantifiers at the beginning consists of alternating blocks of existential and universal quantifiers.  If there are $k$ such blocks in all, and the leftmost block is existential, then we say that the formula is a $\Sigma_k$ formula; if the leftmost block is universal, then the formula is a $\Pi_k$ formula.  We denote the class of such formulas in $FO[<]$ by $\Sigma_k[<]$.

We are  interested in how \fotwobet sits inside $FO[<]$.  One way to measure the complexity of a language in $FO[<]$ is by the smallest number of alternations of quantifiers required in a defining formula, that is the smallest $k$ such that the language is definable by a boolean combination of sentences of $\Sigma_k[<]$.  We will call this the \textit{alternation depth} of the language.  (This is closely related to the \textit{dot-depth}, which can be defined the same way, but with slightly different base of atomic formulas.)
We stress that the alternation depth is measured with respect to arbitrary $FO[<]$ sentences, and not the variable-restricted sentences of \fotwobet.


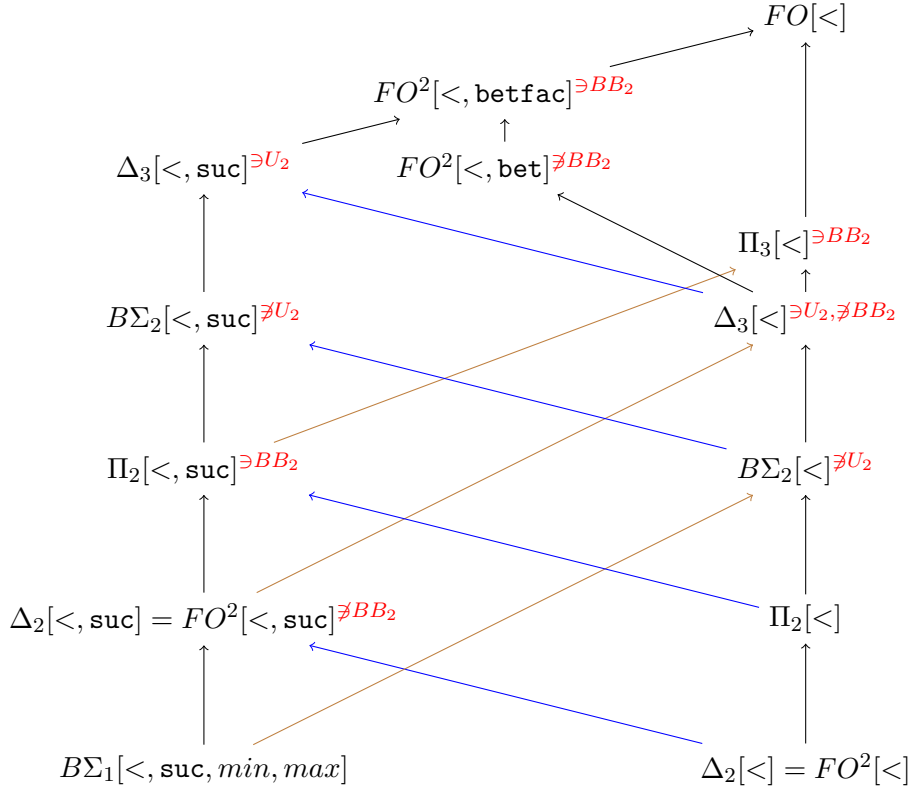
\begin{figure}[t]
\begin{tikzpicture}[scale=1.0,transform shape]

\draw(10,0) node (A) {$\Delta_2[<]=FO^2[<]$};
\draw(10,2) node (C) {$\Pi_2[<]$};
\draw(10,4) node (D) {$B\Sigma_2[<]^{\color{red}\not\ni U_2}$};

\draw(10,6) node (E) {$\Delta_3[<]^{\color{red}\ni U_2, \not\ni BB_2}$};
\draw(10,7) node (G) {$\Pi_3[<]^{\color{red}\ni BB_2}$};
\draw(10,10) node (H) {$FO[<]$};

\draw(2,0) node (Z) {$B\Sigma_1[<,\suc,min,max]$};
\draw(2,2) node (I) {$\Delta_2[<,\suc] = FO^2[<,\suc]^{\color{red}\not\ni BB_2}$};
\draw(2,4) node (K) {$\Pi_2[<,\suc]^{\color{red}\ni BB_2}$};
\draw(2,6) node (L) {$B\Sigma_2[<,\suc]^{\color{red} \not\ni U_2}$};
\draw(2,8) node (X) {$\Delta_3[<,\suc]^{\color{red}\ni U_2}$};

\draw(6,8) node (M) {$FO^2[<,\betw]^{\color{red} \not\ni BB_2}$};
\draw(6,9) node (W) {$FO^2[<,\betfac]^{\color{red}\ni BB_2}$};

\draw [->](A) edge[blue] (I);
\draw [->] (A)-- (C);
\draw [->](C)-- (D);
\draw [->](Z)edge[brown] (D);
\draw [->](I)edge[brown] (E);
\draw [->](K) edge[brown] (G);

\draw [->] (E)-- (G);
\draw [->](G)-- (H);

\draw [->](Z)-- (I);
\draw [->] (I)-- (K);
\draw [->](K)-- (L);
\draw [->](L)-- (X);
\draw [->](X)-- (W);

\draw [->](D)-- (E);
\draw [->](C) edge[blue] (K);
\draw[->](D) edge[blue] (L);
\draw [->](E) edge[blue] (X);
\draw [->](W)-- (H);

\draw [->](E)-- (M);
\draw [->](M)-- (W);

\end{tikzpicture}
\caption{Dot depth and quantifier alternation hierarchies.
The language {\color{red} $U_2$} over alphabet $A=\{a,b,c\}$
is $(A^*\setminus(A^*ac^*aA^*)) \cup (A^*\setminus(A^*bc^*bA^*))ac^*aA^*$,
it consists of words which have
no occurrence of $bc^*b$ before an occurrence of $ac^*a$.
The language {\color{red} $BB_2$} over $\{a,b\}$ is ${(a{(ab)}^*b)}^*$.
}%
\label{fig:hierarchies}
\end{figure}


For the reader familiar with the lower levels of the
quantifier alternation hierarchy of first-order logic
(see~\cite{DGK} for a survey),
these are the classes on the right in Figure~\ref{fig:hierarchies}.
Those on the left are the classes of the original dot depth
hierarchy of Cohen and Brzozowski~\cite{CB}.
The logics which we have introduced in~\cite{KLPS,KLPS18} and in this paper
are at top centre. They have a nonempty
intersection with every level of both the hierarchies.
The two example languages $BB_2$ and $U_2$ (defined in the diagram caption) have played a prominent role in our work.

\begin{thm}\label{thm.alternationdepth}
 The alternation depth of languages in \fotwobet is unbounded.
\end{thm}

\begin{proof}

Consider an alphabet consisting of the symbols
\[0,1,\vee_1,\wedge_2,\vee_3,\wedge_4,\ldots.\]
We define a sequence of languages by regular expressions as follows:
\[C_1=\vee_1{(0+1)}^+\]
\[T_1=\vee_1{(0+1)}^*1{(0+1)}^*.\]
For even $m>1$,
\[C_m=\wedge_m C_{m-1}^+,T_m=\wedge_m T_{m-1}^+.\]
For odd $m>1$,
\[C_m=\vee_m C_{m-1}^+, T_m=\vee_m C_{m-1}^*T_{m-1}C_{m-1}^*.\]

Observe that $C_m$ and $T_m$ are languages over a finite alphabet of $m+2$ letters.  $C_m$ denotes the set of prefix encodings of depth $m$ boolean circuits with 0's and 1's at the inputs.  In these circuits the input layer of 0's and 1's is followed by a layer of unbounded fan-in OR gates, then alternating layers of unbounded fan-in AND and OR gates (strictly speaking, these circuits are trees of AND and OR gates).  $T_m$ denotes the set of encodings of those circuits in $C_m$ that evaluate to \true.

We will show by induction  that for all $m$, both $T_m$ and $C_m$ are definable in $FO^2[<,bet]$, and then show that the alternation depth of the languages $C_m$ and $T_m$ grows without bound as $m$ increases.

 We obtain a sentence  defining $C_1$ by saying that the first symbol is $\vee_1$ and every symbol after this is either 0 or 1.
\[\exists x( \vee_1(x) \wedge \forall y(x\leq y \wedge x<y\rightarrow (1(y) \vee 0(y)))).\]
We obtain a sentence for $T_1$ by taking the conjunction of this sentence with $\exists x 1(x)$.

For the inductive step, we suppose that we have a sentence defining $T_k$ for $k\geq 1$.  Let's suppose first that $k$ is odd.  Thus $T_k$ is a union of $\equiv_r$-classes for some $r$, which we assume to be at least 2. Consequently, whenever $v\in T_k$ and $v'\notin T_k$, Player 1 has a winning strategy in the $r$-round game in $v$ and $v'$.  The proof now proceeds by showing that whenever $w\in T_{k+1}$ and $w'\notin T_{k+1}$, Player 1 has a winning strategy in the $(r+1)$-round game in these two words.  This implies that the $\equiv_r$-class of $w$ is contained within $T_{k+1}$, and thus $T_{k+1}$ is a union of such classes, and hence definable by a sentence in our logic.

First the easy cases: If $w'$ has no occurrence of $\wedge_{k+1}$, then Player 1 wins the game in one round.  If $w'$ has two or more occurrences of $\wedge_{k+1}$, then Player 1 wins the game in two rounds, by first playing on the unique $\wedge_{k+1}$ in $w$ and then jumping in $w'$ to a second occurrence.  If $\wedge_{k+1}$ is not the first letter of $w'$, then Player 1 likewise wins in two rounds by playing first on this occurrence, and then moving the pebble in $w'$  one position left in the second round. If $\vee_{k}$ is not the second letter of $w'$, then Player 2 wins in two rounds again, by playing first on  the initial letter of $w$, and then moving this pebble in the second round to the second position of $w$.  Player 2 must respond by moving the pebble in $w'$ one position to the the right as well, because the set of letters in jumped positions must be empty.

Thus if Player 2 is to have a chance at all, $w'$ must have a factorization
\[w'=\wedge_{k+1} \vee_k w'_1\cdots \vee_k w'_p,\]
where each $w'_1$ contains no occurrences of either $\vee_k$ or $\wedge_{k+1}$, and $p\geq 1$.
Let
\[w=\wedge_{k+1} \vee_k w_1\cdots \vee_k w_q\]
be the corresponding factorization of $w$.  Since $w\in T_{k+1}$, every factor $\vee_k w_j$ belongs to $T_k$.  However, some factor $\vee_k w'_s$ is not in $T_k$.  So Player 1 opens by putting the pebble on the first letter of the factor $\vee_k w'_s$, and Player 2 must respond on the first letter of $\vee_k w_t$ for some $t$.  By the inductive hypothesis, Player 1 has a winning strategy in the $r$-round game in $\vee_k w'_s$ and $\vee_k w_t$.  Player 1 now proceeds to  play this strategy.  Observe that every move he makes inside $\vee_k w_t$ or $\vee_k w'_s$ must be answered my a move within one of these factors, since Player 2 is not allowed to jump over an occurrence of $\vee_k$.   Thus Player 1 wins the $r$-round game in $w,w'$.

The argument for the case where $k$ is even, and for the languages $C_k$, is identical.  Observe that since $T_1$ is defined by a formula of quantifier depth 2, we can take $r=k+1$.

It remains to prove the claim about unbounded alternation depth.  It is possible to give an elementary proof  of this using games. However, by deploying some more sophisticated results from circuit complexity, we can quickly see that the claim is true. Let us suppose that we have a language $L\subseteq{\{0,1\}}^*$ recognized by a constant-depth polynomial-size family of unbounded fan-in boolean circuits; that is, $L$ belongs to the circuit complexity class $AC^0$.  We can encode the pair consisting of a word $w$ of length $n$ and the circuit for length $n$ inputs by a word $p(w) \in C_n$. We now  have $w\in L$ if and only if $p(w)\in T_n$.

Now if the alternation depth of all the $T_n$ is bounded above by some fixed integer $d$,  then we can recognize every $T_n$ by a polynomial-size family of circuits of depth $d$.  We can use this to obtain a polynomial family of circuits of depth $d$ recognizing $L$.  This contradicts the fact (see Sipser~\cite{Sipser}) that the required circuit depth of languages in $AC^0$  is unbounded.
\end{proof}

The situation appears to be different if we bound the size of the alphabet.  In an earlier paper~\cite{KLPS}, we showed that if $|A|=2$, then the alternation depth of languages in $A^*$ definable in $\fotwobet$ is bounded above by 3.  We conjecture that for any fixed alphabet $A$, the alternation depth is bounded above by a linear function of $|A|$, but this question remains open.


\section{Temporal logics and satisfiability}\label{sec:tlsat}
\newcommand{\oomit}[1]{}

We denote by $\utl$ temporal logic with two operators $\fut$ and $\past$. ($\fut$ and $\past$ stand for `future' and `past'.)
 Atomic formulas are propositions
$p\in PV$, a finite set of propositions.  Formulas are built from atomic formulas by applying the boolean operations $\wedge,\vee$, and $\neg$, and the modal operators $\gamma\mapsto \fut\gamma$, $\gamma\mapsto\past\gamma$.

Temporal logics have been interpreted over infinite as well as finite words. Here we confine ourselves only to finite words.
So the propositions $PV$ can be just the alphabet $A$.
We interpret these formulas in marked words. Thus $(w,i)\models a$ if $w(i)=a$, where $w(i)$ denotes the $i^{th}$ letter of $w$. Boolean operations have the usual meaning. We define $(w,i)\models \fut\gamma$ if there is some $j>i$ such that $(w,j)\models\gamma$, and $(w,i)\models\past\gamma$ if there is some $j<i$ with $(w,j)\models\gamma$.

In the setting of verification, it is convenient working with words over a \emph{set} of propositions $\wp(PV)$
as the alphabet, and a boolean formula determining the set which holds at a position $i$.
We also find it convenient to use a set of propositions in Section~\ref{sec:twosat}
where we have to perform reductions between problems.

We can also interpret a formula in ordinary, that is, unmarked words, by defining $w\models\gamma$ to mean $(w,1)\models\gamma$.  Thus temporal formulas, like first-order sentences, define languages in $A^*$.  The temporal logic $\utl$ is known to define exactly the languages definable in $FO^2[<]$ (see~\cite{EVW,TW}).

The \emph{size} of a temporal logic formula can be measured as usual,
using the parse tree of the formula, or
using the parse \emph{dag} of the formula, where a syntactically unique subformula
occurs only once.
The latter representation is shorter and is used in the
formula-to-automaton construction~\cite{VW}.
The complexity of decision procedures for model checking and
satisfiability do not change (see~\cite{DGL}).
For our results, we will use the latter notion of \emph{dag-size}.

\subsection{Betweenness of letters}

We now define new temporal logics by modifying the modal operators $\fut$ and $\past$ with \textit{betweenness and threshold constraints}---these are versions of the between relations $\diam{u}(x,y)$ and $(a,k)(x,y)$ that we introduced earlier.   Let $B\subseteq A$.  A threshold constraint is an expression of the form $\#B\sim c$,  where $c\geq 0$, and $\sim$ is one of the symbols $\{<,\leq,>,\geq,=\}$. Let  $w\in A^*$ and $1\leq i<j\leq |w|$. We say that $(w,i,j)$ satisfies the threshold constraint $\#B\sim c$ if
\[|\{k: i<k<j\text{ and }w(k)\in B\}|\sim c.\]
We can combine threshold constraints with boolean operations $\wedge,\vee,\neg$.  We define satisfaction of a boolean combination of threshold constraints in the obvious way---that is, $w(i,j)$ satisfies $g_1\vee g_2$ if and only if $(w,i,j)$ satisfies  $g_1$ or $g_2$, and likewise for the other boolean operations.

If $g$ is a boolean combination of threshold constraints, then our new operators $\gfut{g}$ and $\gpast{g}$ are defined as follows: $(w,i)\models \gfut{g}\gamma$ if there exists $j>i$ such that $(w,i,j)$ satisfies $g$ and $(w,j)\models\gamma$, $(w,i)\models\gpast{g}\gamma$ if and only if there exists $j<i$ such that $(w,j,i)$ satisfies $g$ and $(w,j)\models\gamma$.

\begin{exas}
\noindent We can express $\fut\gamma$ with threshold constraints as $\gfut{\#\emptyset=0}\gamma$.

We use $\nextt$ to denote the `next' operator:  $(w,i)\models \nextt\gamma$ if and only if $(w,i+1)\models\gamma$. We can express this with threshold constraints by $\gfut{\#A=0}\gamma$.

We can define the language ${(ab)}^+$ over the alphabet $\{a,b\}$ as the conjunction of several subformulas: $a\wedge \nextt b$ says that the first letter is $a$ and the second $b$.  $\neg \fut (a\wedge \nextt  a)$ says that no  occurrence of $a$ after the first letter is immediately followed by another $a$, and similarly we can say that no occurrence of $b$ is followed immediately by another $b$. The formula $\fut(b\wedge\neg\nextt (a\vee b))$ says that the last letter is $b$.
\end{exas}

We denote by  \bthtl temporal logic with these modified operators $\gfut{g}$ and $\gpast{g}$, where $g$ is a boolean combination of threshold constraints.
We also define several fragments of  \bthtl.
In \thtl we further restrict guards $g$ to be atomic, rather than boolean combinations.
In \invtl we restrict \thtl to constraints of the form $\#B=0$---we call these \textit{invariance constraints} since they assert invariance of absence of letters of $B$ in parts of words.
In \binvtl we allow boolean combinations of such invariance constraints.

\begin{exa}
It is useful to have boolean combinations of threshold constraints.
The language $STAIR_k$ given in Section~\ref{sec:twotwo} can be defined by
$\fut (\gfut{\#a=k \land \#b=0} ~true)$.
\end{exa}

\begin{thm}\label{thm.tlequivalence} The logics \thtl, \bthtl, \invtl, \binvtl,
\fotwothr and \fotwobet all define the same family of languages.
\end{thm}
\begin{proof}
In terms of syntactic classes, we obviously have
\[\invtl\subseteq\binvtl\subseteq\bthtl,\]
and
\[\invtl\subseteq\thtl\subseteq\bthtl.\]

We show that $\bthtl\subseteq\invtl$ using a chain of inclusions
passing through the two-variable logics.
In performing these translations, we will only discuss the future modalities,
since the past modalities can be treated the same way.

We can directly translate any formula in \bthtl into an equivalent formula of \fotwothr
with a single free variable $x$:  A formula $\gfut{g}\psi$,
where $g$ is a boolean combination of threshold constraints,
is replaced by a quantified formula $\exists y (y>x \wedge \alpha\wedge \beta(y))$,
where $\alpha$ is a boolean combination of formulas $(a,k)(x,y)$ and $\beta$ is the translation of $\psi$.

We know from Theorem~\ref{thm.invequalsthr} that any formula of \fotwothr can in turn be translated
into an equivalent formula of \fotwobet. Furthermore, \fotwobet is  equivalent to \binvtl:
A simple game argument shows that equivalence of marked words with respect to \binvtl formulas
of modal depth $k$ is precisely the relation $\equiv_k$ of equivalence with respect to
\fotwobet formulas with quantifier depth $k$.

So it remains to show that $\binvtl\subseteq \invtl$:
We do this by translating $\gfut{g}\phi$,
where $g$ is a boolean combination of constraints of the form $\#B=0$,
into a formula that uses only single constraints of this form.

We can rewrite the boolean combination as the disjunction of conjunctions of constraints
of the form $\#\{a\}=0$ and $\#\{a\}>0$.
Since, easily, $\gfut{g_1\vee g_2}\gamma$ is equivalent to $\gfut{g_1}\gamma \vee\gfut{g_2}\gamma$,
we need only treat the case where $g$ is a conjunction of such constraints.
We  illustrate the  general procedure  for translating such conjunctions with an example.
Suppose $A$ includes the letters $a,b,c$.
How do we express $\gfut{g}\gamma$, where $g$ is $(\#\{a\}=0)\wedge(\#\{b\}>0)\wedge(\#\{c\}>0)$?
The letters $b$ and $c$ must appear in the interval between the current position
and the position where $\gamma$ holds.  Suppose that $b$ appears before $c$ does.  We write this as
\[\gfut{\#\{a,b,c\}=0}(b\wedge\gfut{\#\{a,c\}=0}(c\wedge\gfut{\#\{a\}=0}\gamma)).\]
We take the disjunction of this with the same formula in which the roles of $b$ and $c$ are reversed.
\end{proof}
Finally, we remark that $\gfut{\#\{a_1, \ldots, a_p\}=0} \gamma ~~\equiv~~ \bigwedge_{1 \leq i \leq p} ~ \gfut{\#\{a_i\}=0\}} \gamma$, and similarly for
the past modality. Lemma~\ref{lem:singlepair} in the next section gives a proof of this. Thus, \invtl can be simplified to consider singleton sets only.

\subsection{Betweenness of factors}

The temporal logics of previous section allowed $\fut$ and $\past$ modalities to be constrained by threshold counting constraints
on occurrence of letters. In this section, we extend this to threshold counting constraints on occurrence of factors (words).
For example, a formula $F_{\{(\#ab \geq 3) \land (\#c=0)\}} ~~c$ states that there is next occurrence of $c$ in future with at least 3 occurrences of factor $ab$ in-between.
We denote by \bthfactl, the temporal logic with these modified operators $\gfut{g}$ and $\gpast{g}$, where $g$ is a boolean combination of threshold factor constraints.  We also define several fragments of \bthfactl:  For example, \bfactl has the modalities
$\gfut{R}\gamma$ and $\gpast{R}\gamma$ where $R$ is a boolean combination of \emph{requirements on occurrences of factors}. Since disjunctions can be pulled out, we write the conjunctive requirements as a finite set of positive and negative  factors. Thus, for simplicity, $R$ has the form $\{u_1,\ldots,u_p, \neg v_1, \ldots, \neg v_r\}$ in $\bfactl$. This $R$ denotes the constraint $(\#u_1 > 0 \land \ldots \#u_p > 0 \land \#v_1=0 \land \ldots \#v_r = 0)$. Logic \factl is a further restricted version where there are no positive requirements and a single  negative requirement is permitted, i.e. $R$ has the form $\{\neg v\}$.

Formally, we have $(w,i)\models \gfut{R}\gamma$ if for some $j>i$, $(w,j)\models \gamma$
and for every positive factor $u \in R$ there is an occurrence of factor $u$
between, that is, for some $i < k < k+|u|-1 < j$, $(w,k)\models st(u)$,
and for every negative factor $\lnot v \in R$ there is no occurrence of
factor $v$ between,
that is, for all $i < k < k+|v|-1 < j$, $(w,k)\not\models st(v)$.

\medskip
\noindent\textit{Notation.}
For convenience,
for $u = a_1 a_2\dots a_n$, we will abbreviate by $st(u)$ the \ltl formula
$a_1 \land \nextt(a_2 \land (\dots \land \nextt a_n))$ and by
$end(u)$ the formula $a_n \land \prev(a_{n-1} \land (\dots \land \prev a_1))$.
Also, for $u=\epsilon$ (empty string) define $st(u)=end(u)=true$.

\begin{exas}
The formula $\henceforth(st(bb) \limplies \past(end(aa)))$ specifies that
there every occurrence of factor $bb$  is preceded by an $aa$.

$\henceforth[st(bb) \land \past(end(bb)) \limplies \gpast{\{aa, \neg bb \}}(end(bb))]$ says
that between any two occurrences of $bb$ there must be an $aa$.
\end{exas}

\begin{thm}\label{thm.factlequivalence}
The logics \bthfactl, \bfactl, \factl, \fotwothrfac and \fotwobetfac all define the same family of languages.
\end{thm}

As in the proof of Theorem~\ref{thm.tlequivalence}, the key question is how to translate
\bfactl into \factl. All other reductions are straightforward extensions (substituting factors in place of letters) of those given in the previous proof. Hence we omit repeating these.
For showing $\bfactl \subseteq \factl$, We consider the modality
$\gfut{\{u_1,\ldots,u_p, \neg v_1, \ldots, \neg v_r\} } \gamma$  and construct an equivalent
formula using only the modality of the form $\gfut{\{\neg v\}} \gamma$. Note that the size of the former formula is taken to be
$(|u_1| + \cdots + |u_p| + |v_1| + \cdots + |v_r|) + size(\gamma) +1$.

\begin{lem}\label{lem:singlepair}
$\gfut{\{u_1,\ldots,u_p, \neg v_1, \ldots, \neg v_r\} } \gamma \quad \equiv \quad
 \bigwedge_{1 \leq i \leq p} \bigwedge_{1 \leq j \leq r} \gfut{\{u_i,\neg v_j\}} \gamma$.
 The dag-size of the right hand formula is polynomial in the size of the left hand formula.
\end{lem}
\begin{proof}
First note that
\[
\gfut{\{u_1,\ldots,u_p, \neg v_1, \ldots, \neg v_r\} } \gamma \quad \equiv \quad
 \bigwedge_{1 \leq i \leq p} \gfut{\{ u_i, \neg v_1, \ldots, \neg v_r\} } \gamma
\]
The forward direction is clear. For the converse, let $k$ be the position where the right hand
conjunction holds. Let $x_i > k$ such that
$(w,x_i)\models \gamma$ witnessing $(w,k)\models \gfut{\{ u_i, \neg v_1, \ldots, \neg v_r\} } \gamma$.
Consider $x=\max\{x_i\}$. Then, $(w,x)\models \gamma$ witnessing the left hand formula.

Also,
\[
\gfut{\{ u, \neg v_1, \ldots, \neg v_r\} } \gamma \quad \equiv \quad
 \bigwedge_{1 \leq j \leq r} \gfut{\{u,\neg v_i\}} \gamma
\]
The forward direction is clear. For the converse, let $k$ be the position where the right hand
conjunction holds. Let $x_i > k$ s.t.
$(w,x_j)\models \gamma$ witnessing $(w,k)\models \gfut{\{ u, \neg v_j\}}$.
Consider $x=\min\{x_i\}$. Then, $(w,x)\models \gamma$ witnessing the left hand formula.

The result follows by combining the two steps. It is also clear that resulting formula has $p.r$ conjuncts of constant size which each use $\gamma$ as subformula. Hence, the dag-size of the conjunction is polynomial
in the size of the left hand formula.
\end{proof}

Because of above Lemma~\ref{lem:singlepair}, we only need to consider formulas with modalities of the form
$\gfut{\{u,\neg v\}} \gamma$  and $\gpast{\{u,\neg v\}} \gamma$.
Next, we will reduce $\gfut{\{u,\neg v\}} \gamma$ to an equivalent formula only using the modalities of the form
$\gfut{\{\neg v\}} \gamma$  which we abbreviate as $\gfut{\neg v}\gamma$.
A similar reduction can be carried out for the past modality too.

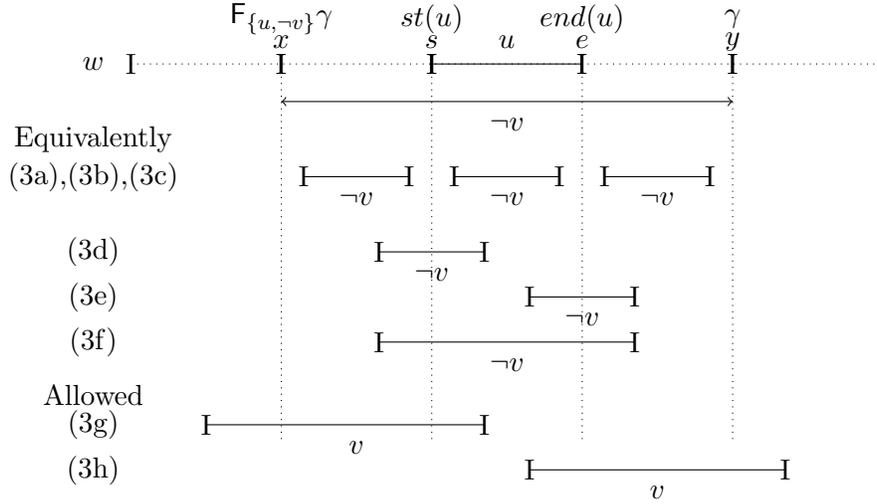
\begin{figure}
\begin{tikzpicture}[scale=1.0,transform shape]

\draw (0.5,5) node {$w$};
\draw(1,5.3) node{};
\draw[style=dotted] (1,5) node{I} -- (11,5);
\draw (3,5) node{I};\draw(3,5.3) node{$x$};
\draw (3,5) node{I};\draw(3,5.6) node{$\gfut{\{u,\lnot v\}}\gamma$};
\draw (5,5) node{I};\draw(5,5.3) node{$s$};
\draw (5,5) node{I};\draw(5,5.6) node{$st(u)$};
\draw (7,5) node{I};\draw(7,5.3) node{$e$};
\draw (7,5) node{I};\draw(7,5.6) node{$end(u)$};
\draw (9,5) node{I};\draw(9,5.3) node{$y$};
\draw (9,5) node{I};\draw(9,5.6) node{$\gamma$};

\draw[style=dotted] (3,5) --(3,0);
\draw[style=dotted] (5,5) --(5,0);
\draw[style=dotted] (7,5) --(7,0);
\draw[style=dotted] (9,5) --(9,0);

\draw[style=solid] (5,5) node{I} -- (7,5);
\draw (6,5.3) node{$u$};

\draw[<->] (3,4.5) -- (9,4.5);
\draw (6,4.2) node{$\neg v$};
\draw (0.5,4) node{Equivalently};

\draw (0.5,3.5) node{(3a),(3b),(3c)};
\draw (3.3,3.5) node{I}  -- (4.7,3.5) node{I}; \draw (5.3,3.5) node{I}  -- (6.7,3.5) node{I};  \draw (7.3,3.5) node{I}  -- (8.7,3.5) node{I};
\draw (4,3.2) node{$\neg v$}; \draw (6,3.2) node{$\neg v$}; \draw (8,3.2) node{$\neg v$};

\draw (0.5,2.5) node{(3d)};
\draw (4.3,2.5) node{I} -- (5.7,2.5) node{I}; 
\draw (5,2.2) node{$\neg v$};

\draw (0.5,1.9) node{(3e)};
\draw (6.3,1.9) node{I} -- (7.7,1.9) node{I}; 
\draw (7,1.6) node{$\neg v$};

\draw (0.5,1.3) node{(3f)};
\draw  (4.3,1.3) node{I} -- (7.7,1.3) node{I};
\draw  (6,1.0) node{$\neg v$};


\draw  (0.5, 0.6) node{Allowed};
\draw (0.5,0.2) node{(3g)};  \draw (2,0.2) node{I} -- (5.7,0.2) node{I};
\draw (4.0,-0.1) node{$v$};

\draw (0.5,-0.4) node{(3h)}; \draw(6.3,-0.4) node{I} -- (9.7,-0.4) node{I};
\draw(8,-0.7) node{$v$};
\end{tikzpicture}
\caption{$w,x \models \gfut{\{u,\neg v\}} \gamma$ with allowed and forbidden occurrences of $v$,
classified as (3a)--(3f).}\label{fig:overlap} 
\end{figure}


We have to consider several cases which are illustrated in Figure~\ref{fig:overlap}.
Notice that, from the semantics, $(w,x)\models \gfut{\{u,\neg v\}} \gamma$ {\bf iff} for given $w,x$ there exist
$s,e,y$ such that $x < s < e < y$ and {\bf (1)} $w[s,e]=u$, {\bf (2)}~$(w,y)\models \gamma$, and {\bf (3)}~$v$ is not a factor of $w[x+1,y-1]$. Also, notice that, without loss of generality, it is sufficient to consider the earliest $s,e,y$ satisfying~(1),~(2) and to check for~(3). We shall assume this henceforth.

The above condition~(3) that  $v$ is not a factor of $w[x+1,y-1]$ is equivalent to the conjunction of conditions: $v$ is not a factor of any of {\bf (3a)}~$w[x+1,s-1]$, {\bf (3b)}~$w[s,e]$, {\bf (3c)}~$w[e+1,y-1]$,
as well as the three overlap conditions below.

\noindent
{\bf (3d')}~$v$ does not overlap $u$ at $s$, more precisely:
$v=v_1u_1 ~\mbox{and}~ u=u_1u_2 ~\mbox{with}~ u_1,u_2,v_1 \neq \epsilon$. We call $v_1$ a \emph{pre-overlap}
and let $Opre_1$ be the set of possible pre-overlaps.

\noindent
{\bf (3e')}~$v$ does not overlap $u$ at $e$, more precisely:
$u=u_1u_2 ~\mbox{and}~ v=u_1v_2 ~\mbox{with}~ u_1,u_2,v_2 \neq \epsilon$. We call $v_2$ a \emph{post-overlap}
and let $Opost_1$ be the set of all possible post-overlaps.

\noindent
{\bf (3f')}~$v$ does not overlap $u$ at both $s$ and $e$, that is,
$v = v_1 \cdot u \cdot v_2$ with $v_1 \neq \epsilon \lor v_2 \neq \epsilon$. We call $v_1$ a
\emph{pre-super-overlap} and $v_2$ a \emph{post-super-overlap}.
In this case we let $Opre_2=\{v_1\}$ and $Opost_2=\{v_2\}$ where $v_1,v_2$ are obtained by the  factorization of $v$
which makes $v_1$ the shortest.
If no such pair $(v_1,v_2)$ exists, let $Opre_2=Opost_2=\emptyset$.
Thus, $Opre_2,Opost_2$ are either both singleton sets or both are empty.

Note that, in condition~(3f') above, we have considered only the factorization $v=v_1 \cdot u \cdot v_2$ which makes $v_1$ the shortest
such factor. (There may be multiple factorizations of $v$ possible. See the example below.) This is justified as we are considering the earliest occurrence of $u$, and any other factorization will lead to an earlier occurrence of $u$.

\begin{exas}
For $u=aaa$ and $v=bbb$, trivially, $Opre_1=Opost_1=Opre_2=Opost_2=\emptyset$ since there are no overlaps.

For $u=ababbb$ and $v=bbabab$, because of pre-overlaps
$(v=bb \cdot abab,u=abab \cdot bb)$ and $(v=bbab \cdot ab,u=ab \cdot ab bb)$, we get $Opre_1=\{bb,bbab\}$. Note that which of these overlaps will
apply will depend of word $w$. If $w=bbabababbb$ then $Opre_1$ element $bbab$ will
characterise pre-overlap  whereas
if $w=bbababbb$ then $Opre_1$ element $bb$ will characterize the pre-overlap.
Also, due to post-overlap
$(u=abab \cdot bb,v=bb \cdot abab)$ we have $Opost_1= \{ abab \}$. Since,
$u$ is not subword of $v$, we have $Opre_2=Opost_2=\emptyset$.

For $u=aa$ and $v=aaaabbaaaa$, we have several possible overlaps. The only pre-overlap is  $(v=aaaabbaaa \cdot a,u=a \cdot a)$. Super-overlaps are  $v=aaaabbaa \cdot aa \cdot$, $aaaabba \cdot aa \cdot a$,
$aaaabb \cdot aa \cdot aa$, $aa \cdot aa \cdot bb aa aa$, $ a \cdot aa \cdot abbaaaa$, $\cdot aa \cdot aa bb aaaa$.
We have one post-overlap $(v=a \cdot aaabbaaaa,u= a \cdot a)$. Hence, we
get $Opre_1=\{ aaaabbaaa\}$, $Opost_1=\{aaabbaaaa\}$. It is sufficient to consider the canonical super-overlap $\cdot aa \cdot aabbaaaa$ giving $Opre_2=\{ \epsilon \}$ and $Opost_2=\{ aabbaaaa\}$.
\end{exas}

Observe that we can rephrase conditions (3d'-3f') to equivalent conditions (3d-3f) as follows.

\noindent
{\bf (3d)} for all pre-overlaps $v' \in Opre_1$ we have $v'$ is not a suffix of $w[x+1,s-1]$.

\noindent
{\bf (3e)} for all post-overlaps $v' \in Opost_1$ we have $v'$ is not a prefix of $w[e+1,y-1]$.

\noindent
{\bf (3f)} Since $Opre_2=Opost_2 = \emptyset$ or $Opre_2=\{v_1'\}, Opost_2=\{v_2'\}$, (3f)
 states that $Opre_2= \emptyset$, or, pre-super-overlap $v_1'$ being suffix of $w[x+1,s-1]$ implies
post-super-overlap $v_2'$ is not a prefix of $w[e+1,y-1]$. That is,
$v_1'uv_2'$ does not super-overlap $u$.

\begin{lem}\label{lem:formcond}
 $w,x \models \gfut{\{u,\neg v\}} \gamma$ {\bf iff} for given $w,x$ there exist $s,e,y$ s.t.
 $x<s<e<y$ and conditions~(1),~(2),~(3a-3f) are satisfied. \qed
\end{lem}

We now define a formula $\beta \in \factl$ to enforce conditions~(1),~(2),~(3a-3f).
By the previous lemma, this will give $\beta$ equivalent to $\gfut{\{u,\neg v\}} \gamma$.
If $u=u_1 \cdot v \cdot u_2$, then condition (3b) is violated, and  trivially
$\gfut{\{u,\neg v\}} \gamma \equiv \false$.
In the rest of the section, we assume that $v$ is not a factor of $u$, and hence condition (3b) is always satisfied.

\bigskip

\noindent For $U,V \subseteq_{fin} A^*$, define as below, with
empty disjuncts vacuously evaluating to $\false$:
\begin{quote}
 $Pre(U) = \bigvee_{(v' \in U)} ~\prev end(v')$, \quad \quad
 $Pre(U,i) = \bigvee_{(v' \in U, |v'|<i)}  ~\prev end(v')$ \\
 $Post(V) = \bigvee_{(v' \in V)} ~ \nextt st(v')$, \quad \quad
 $Post(V,i) = \bigvee_{(v' \in V, |v'|<i)} \nextt st(v')$
\end{quote}

Then,  $(w,s)\models \neg Pre(U)$ states that no $v' \in U$ ends at $s-1$  in $w$.
Moreover, $(w,s)\models \neg Pre(U,i)$ states that no $v' \in U$ of length shorter than $i$ ends at $s-1$.
Formulas $\neg Pre(U)$ and $\neg Pre(U,i)$ will be used to enforce condition (3d). Similarly, formulas
$\neg Post(V)$ and $\neg Post(V,i)$ will be used to enforce condition (3e).

\medskip

For given $u,v$, let $Opre_1,Opost_1,Opre_2,Opost_2$, as well as predicates
$Pre(\cdot), Post(\cdot)$ on them, be syntactically determined as defined above.
These predicates have to be applied at different points.
We define a subsidiary formula $\delta(V)$ with a parameter
$V \subseteq_{fin} A^*$, to check the post-requirements.
Let $m= \max\{|v'| ~\mid~ v' \in V\}$.
\[
\delta(V) = [ (\gfut{\neg v} ~\gamma) ~~\land ~~\neg Post(V)]
    \quad \lor \quad \bigvee_{i=1}^m ~  [ (\nextt^i\gamma) \land \neg Post(V,i)]
\]

\begin{lem}\label{lem:deltacond}
Let $V$ be a given finite set of words such that $m < |v|$ where $\delta$ and  $m$ are as defined above. Then,
 $w,e \models \delta(V)$ {\bf iff} for given $w,e$ there exists $y$ s.t. $e < y$ and conditions
 (2), (3c)  and (NoPost) hold. Here,
 {\bf (NoPost)}:~~for all $v' \in V$,  $v'$ is not a prefix of $w[e+1,y-1]$.
\end{lem}
\begin{proof}
For the direction from left to right, let $w,e \models \delta(V)$. \\
\noindent
Case 1: $w,e \models (\gfut{\neg v} \gamma) \land \neg Post(V)$. Then by the semantics of $\gfut{\neg v} \gamma$, conditions (2) and (3c) are satisfied. Also, the conjunct $\neg Post(V)$ ensures that any $v' \in V$ is not a prefix of $w[e+1,|w|-1]$ and hence not a prefix of $w[e+1,y-1]$. Thus, (NoPost) is satisfied.

\noindent
Case 2: $w,e \models (\nextt^i\gamma) \land \neg Post(V,i)$ for some $1 \leq i \leq m$. Note that $m<|v|$.
Taking $y=e+i$ we get that (2) is satisfied. Moreover, (3c) is satisfied as interval $(e,y)$ of length $i$ is
too small to accommodate $v$. Finally, conjunct $\neg Post(V,i)$ ensures
absence of any $v' \in V$ with $|v'| < i$ is a prefix of $w[e+1,|w|-1]$ and hence a prefix of $w[e+1,y-1]$.
Also, any $v' \in V$ with $|v'| \geq i$  which starts at $e+1$ will end after $y-1$ and hence it cannot be a prefix of $w[e+1,y-1] $.
This is sufficient to satisfy (NoPost).

\medskip

\noindent
In the other direction, let $w,e,y$ be such that $e<y$ with (2), (3c) and~(NoPost) satisfied.

\noindent
Case 1:   $y-e > m$. Then,  $\gfut{\neg v} \gamma$ follows from (2) and (3c).
By (NoPost), every $v' \in V$ is not prefix of $w[e+1,y-1]$. But as $y-e>m$ and $|v'| \leq m$, we have any $v'$ starting at $e+1$ must end before position $y$. Hence, $w,e \models \neg Post(V)$. Thus, $w,e$ satisfies the first disjunct of $\delta(V)$.

\noindent
Case 2: $y-e \leq m$.
Let $i=y-e$.
Then by (2) we have $w,e+i \models \gamma$ giving $w,e \models X^i \gamma$. Also, by (NoPost), for all
$v' \in V$ with $|v'| <i$ we have $v'$ is not a prefix of $[e+1,y-1]$.
Also, any $v' \in V$ with $|v'| \geq i$  which starts at $e+1$ will end after $y-1$ and hence it cannot be a prefix of $w[e+1,y-1] $.
Hence, $w,e \models \neg Post(V,i)$.
Thus, $w,e$ satisfies the second disjunct of $\delta(V)$.
\end{proof}

\medskip

\noindent We now define $\beta$, parameterized by overlap sets. For given $u,v$,  let $Opre_1$, $Opost_1$, $Opre_2$, $Opost_2$ be as defined earlier.
Let $k = \max\{|v'| ~\mid~ v' \in Opre_1 \cup Opre_2\}$.
The formula $\beta$ is:
\[
\begin{array}{l}
  \phantom{\lor ~~}    \gfut{\neg v} ~[st(u) ~\land~ \neg Pre(Opre_1) \land
	[ \quad [\neg Pre(Opre_2) ~\land~ \nextt^{(|u|-1)} (\delta(Opost_1))] \\
        \hspace*{5.5cm} \lor~ [~
                       \nextt^{(|u|-1)} (\delta(Opost_1 \cup Opost_2))]~~]~~] \\
  \lor \\
  \quad \bigvee_{i=1}^{k} ~\nextt^i
	    [st(u) ~\land~ \neg Pre(Opre_1,i)\land [ \quad [(\neg Pre(Opre_2,i)) ~\land~ \nextt^{(|u|-1)} (\delta(Opost_1))]  \\
               \hspace*{6.5cm} \lor~ [
                 ~\nextt^{(|u|-1)} (\delta(Opost_1 \cup Opost_2)) ]~~]~
            ]
   \end{array}
\]
\begin{lem}\label{lem:nfacred}
 $w,x \models \beta$ {\bf iff} for given $w,x$ there exist $s,e$ such that
$x<s<e$ satisfying conditions~(1), (2), (3a-3f).
\end{lem}
\begin{proof}
There are two cases, $Opre_2=Opost_2=\emptyset$, and $Opre_2=\{v_1\}, Opost_2=\{v_2\}$ for some $v_1,v_2$.
We consider the more complex second case. The first case is similar to a subcase of the second one and omitted.

\medskip

\noindent
For proving the forward implication, assume that $w,x \models \beta$.  Notice that $\beta$ has two outer disjuncts.

\noindent
(Case 1: $w,x$ satisfies the first outer disjunct). By subformula
$w,x \models \gfut{\neg v} [st(u)]$,
condition (1) and (3a) are satisfied. We already assumed that (3b) is satisfied. Further, the subformula $\neg Pre(Opre_1)$ ensures that every $v' \in Opre_1$ does not end at $s-1$ and hence it cannot be a
suffix of $w[x+1,s-1]$.  Hence (3d) is satisfied.

The remaining subformula consists of two sub-disjuncts.
In case the first sub-disjunct holds, we have $w,s \models \neg Pre(Opre_2)$,
which implies $v_1$ is not a suffix of $w[x+1,s-1]$. Hence, (3f) is satisfied.
Moreover, in this case, the formula asserts $w,e \models \delta(Opost_1)$.
By the forward direction of Lemma~\ref{lem:deltacond},
conditions~(2), (3c), and $NoPost(Opost_1)$ hold.  Observe that $NoPost(Opost_1)$  implies~(3e).

If the second sub-disjunct holds,
the formula asserts
$w,e \models \delta(Opost_1 \cup Opost_2)$. By Lemma~\ref{lem:deltacond},
conditions~(2), (3c), and $NoPost(Opost_1 \cup Opost_2)$ hold, the last of which implies~(3e) and~(3f).

\medskip

\noindent
(Case 2: $w,x$ satisfies the second outer disjunct of $\beta$). Then, for some $1 \leq i \leq k$ we have
$w,x \models X^i st(u)$, which by taking $s=x+i$, gives (1). Also $i \leq k < |v|$, giving (3a), as interval $[x+1,s-1]$ is too small to contain $v$.
Further, subformula $\neg Pre(Opre_1,i)$ ensures that every $v' \in Opre_1$ with $|v'| < i$ does not end at $s-1$ and hence it cannot be a suffix of $w[x+1,s-1]$.  If $v' \in Opre_1$ with $|v'| \geq i$ then $v'$
cannot be a suffix of $w[x+1,s-1]$ as the interval $[x+1,s-1]$ is shorter than $|v'|$. Hence (3d) is satisfied.

The remaining subformula consists of two sub-disjuncts. Notice that, by reasoning similar to above,  $v_1$ is suffix of $w[x+1,s-1]$ iff $w,s \models Pre(Opre_2,i)$. Then, the proof that conditions~(2), (3c), (3e) and~(3f)
hold is analogous to the proof in Case 1.

\bigskip

\noindent For proving backward implication, assume that for given $w,x$ there exist $s,e$ such that $x<s<e$ and (1), (2), (3a-3f) are satisfied. We consider two cases $s-x > k$ and $s-x \leq k$.
\medskip

\noindent
(Case 1: $s-x > k$). We will show that first outer disjunct holds. By (1) and (3a), we have $w,x \models \gfut{\neg v} st(u)$. Also, (3d) states that for all $v' \in Opre_1$, we have $v'$ is not a suffix of $w[x+1,s-1]$. But, as
$|v'| \leq k$ and $s-x > k$ we get that $v'$ is not the suffix of $w[0,s-1]$. Hence, $w,s \models \neg Pre(Opre_1)$. By similar reasoning,
$v_1 \in OPre_2$ is a suffix of $w[x+1,s-1]$ iff $w,s \models Pre(Opre_2)$.
Also, right to left direction of Lemma~\ref{lem:deltacond} with conditions (2), (3c), (3e) imply $w,e \models \delta(Opost_1)$.
If $v_1$ is a not suffix of $w[x+1,s-1]$, the first sub-disjunct is satisfied.
If $v_1$ is a suffix of $w[x+1,s-1]$ then, condition (3f) implies that $v_2$ is not a prefix of $w[e+1,y-1]$. In this case,
by the right to left direction of Lemma~\ref{lem:deltacond}, we get $w,e \models \delta(Opost_1 \cup Opost_2)$, satisfying the
second sub-disjunct. Thus, the first outer disjunct holds in both cases.

\medskip
\noindent
(Case 2: $s-x \leq k$). Let $i=s-x$. Then $1 \leq i \leq k$. By (1) we have $w,x \models X^i st(u)$. By (3e), for all $v' \in Opre_1$ with $|v'| < i$, we have $v'$ does not end at $s-1$, giving $w,s \models \neg Pre(Opre_1,i)$. Similarly, $v_1 \in Opre_2$ is not a suffix of $w[x+1,s-1]$  iff $w,s \models \neg Pre(Opre_2,i)$. Also, right to left direction of Lemma~\ref{lem:deltacond} with conditions (2), (3c), (3e) imply $w,e \models \delta(Opost_1)$.
If $v_1$ is not a suffix of $w[x+1,s-1]$, the first sub-disjunct is satisfied.
If $v_1$ is a suffix of $w[x+1,s-1]$ then, condition (3f) implies that $v_2$ is not a prefix of $w[e+1,y-1]$. In this case,
by the right to left direction of Lemma~\ref{lem:deltacond}, we get $w,e \models \delta(Opost_1 \cup Opost_2)$, satisfying the
second sub-disjunct. Thus, the second outer disjunct holds in both cases.
\end{proof}

By combining  Lemmas~\ref{lem:formcond} and~\ref{lem:nfacred},
we have $w,x \models \gfut{\{ u, \neg v\}}\gamma$ iff $w,x \models \beta$.
The result that $\bfactl \subseteq \factl$ follows from this  by first applying Lemma~\ref{lem:singlepair}.
It is easy to see that the dag-size of $\beta$,  with the $\delta$ occurrences substituted out,
is polynomial in the size of $\gfut{\{u,\neg v\}} \gamma$.
This completes the proof of Theorem~\ref{thm.factlequivalence}.

\bigskip

Now we give a reduction from \factl to \ltl.
\begin{lem}\label{lem.toltl}
For $\gfut{\neg v} \gamma$ we have an equivalent \ltl formula:
\[
 (\bigvee_{i=1}^{|v|-1} \nextt^i \gamma)  ~~\lor~~ ( (\neg end(v)) \until \gamma).
\]
\end{lem}
The proof follows from the semantics. Note that the dag-size of the \ltl formula is polynomial in the size of $\gfut{\neg v} \gamma$. \qed

\subsection{Complexity of satisfiability}\label{KamalParitosh}

Given a formula in one of these logics, what is the computational complexity of determining whether it has a model, that is, whether the language it defines is empty or not?  This is the \textit{satisfiability problem}  for the logic.  To determine this, we require some way to measure the size of the input formula.  For formulas containing threshold constraints, we code the threshold value in binary, so that mention of a threshold constant $c$ contributes $\lceil\log_2 c\rceil$ to the size of the formula. Mention of a subalphabet $B$ contributes $|B|$ to the size of the formula.
Our results below hold for bounded and unbounded alphabets,
which may be explicitly specified or symbolically specified by propositions.
Recall that we use the dag-size, that is, the input formula
is represented as a dag (see~\cite{DGL}).

We think that providing such factor requirements in temporal
logic, using our simple notation, improves convenience without sacrificing
complexity of satisfiability. Threshold requirements provide greater
succinctness, but that comes with a computational cost.
\begin{thm}\label{thm.utlsat}
The satisfiability problems for the temporal logics $\bfactl,\binvtl$ (with boolean combination of invariance constraints) and
$\invtl$ are \pspace-complete.
The satisfiability problems for $\bthfactl,~ \bthtl$ (with boolean combination of threshold constraints)
are \expspace-complete.
\end{thm}

\begin{proof}
Satisfiability of \invtl is \pspace-hard~\cite{SC} since it includes
\ltlunsuc.
We observe that \bfactl can be translated into  \ltlbin  (in dag-representation) in polynomial time, as
given in the proof of  Theorem~\ref{thm.factlequivalence}.
Giving the upper bound for the factor logic \bfactl takes more work,
but nearly all of it has been done in Theorem~\ref{thm.factlequivalence},
using which we translate to an \factl formula with
polynomially larger dag-size, which only uses singleton negative
requirements. By Lemma~\ref{lem.toltl} this translation can be continued in polynomial time
into an \ltl formula
which uses the binary \emph{until} operation,
and the decision procedure for \ltl~\cite{SC} completes the proof. Syntactically,
$\invtl \subset \binvtl \subset \bfactl$.
Hence, satisfiability of \bfactl, \binvtl, and \invtl is \pspace-complete~\cite{SC}.

Using a threshold constant $2^n$, written in binary in the formula with size $n$,
the  $2^n$-iterated Next operator $X^{(2^n)}$ can be expressed in \thtl,
so we obtain that its satisfiability is \expspace-hard~\cite{AH}.
By an exponential translation of \bthtl into \binvtl,
(the proof of Theorem~\ref{thm.invequalsthr} shows that such translation exists), or alternately by a polynomial
translation into \cltlbin~\cite{LMP},
its satisfiability is \expspace-complete.
Analogously we can translate from \bthfactl into \bfactl
(for which we rely on Theorem~\ref{thm.invequalsthrfac}).
\end{proof}

\subsection{Satisfiability of two-variable logics}\label{sec:twosat}

\begin{thm}\label{thm.fo2sat}
Satisfiability of the two-variable logics
\fotwobet, \fotwobetfac, \fotwothr and \fotwothrfac is \expspace-complete.
\end{thm}

We begin by reducing the \emph{exponential Corridor Tiling problem}
to satisfiability of \fotwobet.
It is well known that this problem is \expspace-complete~\cite{Furer}.
An instance $M$ of the problem is given by $(T,H,V,s,f,n)$
where $T$ is a finite set of tile types with $s,f \in T$,
the horizontal and vertical tiling relations
$H, V \subseteq T \times T$, and $n$ is a natural number.
A solution of the $2^n$ sized corridor tiling problem is
a natural number $m$ and map $\pi$ from
the grid of points $\{ (i,j) ~\mid~ 0 \leq i < 2^n,~0 \leq j < m\}$
to $T$ such that: \\
$\pi(0,0)=s$, $\pi(2^n-1,m-1)=f$ and for all $i$, $j$ on the grid, \\
$(\pi(i,j),\pi(i,j+1)) \in V$ and $(\pi(i,j),\pi(i+1,j)) \in H$.

\begin{lem}\label{lem.tiling}
Satisfiability of the two-variable logic \fotwobet is \expspace-hard.
\end{lem}
\begin{proof}
Given an instance $M$ as above of a Corridor Tiling problem,
we encode it as a sentence  $\phi(M)$ of size $poly(n)$ with a
modulo $2^n$ counter $C(x)$ encoded serially with $n+1$ letters.
As we already saw in the examples in Section~\ref{sec:examples},
the counter could be laid out using successive $n+1$ letter substrings
to represent $n$ bits (over a subalphabet of size $2$) preceded by a marker.
The marker now represents a tile and a colour from $red$, $blue$ and $green$
(requiring subalphabet size $3|T|$).
Thus the $2^{n} \times m$ tiling is represented by a word of length
$m(n+1)2^n$ over an alphabet of size $3|T|+2$.
The claim is that $M$ has a solution iff $\phi(M)$ is satisfiable.

The sentence  $\phi(M) \in \fotwobet$ is a conjunction of
the following properties.
The key idea is to cyclically use monadic predicates
$red(x), green(x), blue(x)$ for assigning colours to rows.

\begin{itemize}
\item
Each marker position has exactly one tile and one colour.
\item
The starting tile is $s$, the initial colour is $red$,
the initial counter bits read $0^n$, the last tile is $f$.
\item
Tile colour remains same in a row and it cycles in order
$red,green,blue$ on row change.
\item
The counter increments modulo $2^n$ in consecutive positions.
\item
For horizontal compatibility we check:\\
$
\begin{array}{l}
\forall x \forall y(mark(x) \land mark(y) \land \neg mark(x,y)
\limplies \\
 \hspace*{1cm} \orover_{(t_1,t_2) \in H} ~ t_1(x) \land t_2(y))
\end{array}
$.
\item
For vertical compatibility, 
we check that $x,y$ are in adjacent rows by invariance of lack of one colour.
We check that $x$ and $y$ are in the same column by checking that
the counter value (which encodes column number) is the same:
\[ \forall x \forall y ( (x < y) \land
(\neg red(x,y) \lor \neg blue(x,y) \lor \neg green(x,y)) \land EQ(x,y) \limplies
\orover_{(t_1,t_2) \in V} ~ t_1(x) \land t_2(y)).
\]
\end{itemize}

\noindent
It is easy to see that we can effectively translate an instance $M$ of
the exponential corridor tiling problem into $\phi(M)$ in time
polynomial in $n$. The translation preserves satisfiability.
Hence, by reduction, satisfiability of $\fotwobet$ over bounded
as well as unbounded alphabets is $\expspace$-hard.
\end{proof}

\begin{lem}\label{lem.reduction}
There are satisfiability-preserving polynomial time reductions from
the logic \fotwothr to the logic \fotwobet, and from
the logic \fotwothrfac to the logic \fotwobetfac.
\end{lem}
\begin{proof}
We give a polytime reduction from \fotwothr to  \fotwobet
which preserves satisfiability. The same idea works for the reduction
from \fotwothrfac to \fotwobetfac.
We consider in the extended syntax
a threshold constraint  $\#a(x,y)=k$ where $a$ is a letter or a proposition and
$k$ is a natural number.
The key idea of the reduction is illustrated by the following example.

For each threshold constraint $g$ of the form $\#a(x,y)=2^r$,
we specify a \emph{global} modulo $2^r$ counter $C_g$
using monadic predicates
$p_1(x), \ldots, p_r(x)$ whose truth yield an $r$-bit vector denoting  the value of the counter at position $x$.
(This requires a symbolic alphabet
where several such predicates may be true at the same position.)
By ``global'' we mean that the counter $C_g$ has value $0$
at the beginning of the word and it increments whenever $a(x)$ is true.
This is achieved using the formula:
\[
\forall x,y (\suc(x,y) \limplies (a(x) \limplies INC_g(x,y)) \land
 (\neg a(x) \limplies EQ_g(x,y)))
\]
Here, the formulas $EQ_g(x,y)$ and $INC_g(x,y)$ are similar to those of the serial counter example in Section~\ref{sec:examples}.
Also, we have three colour predicates $red_g(x)$, $blue_g(x)$ and $green_g(x)$
where the colour at the beginning of the word is $red_g$,
and we change the colour cyclically each time the counter $C_g$
resets to zero by overflowing.
As in the proof of Lemma~\ref{lem.tiling},
invariance of lack of one colour and the fact that $x,y$ have different colours
ensures that the counter overflows at most once.
We replace the constraint $g$
of the form $\#a(x,y)=2^r$ by an equisatisfiable formula:
\[
\begin{array}{l}
x<y \land EQ_g(x,y) \land \#a(x,y) > 0 \land\\
(\neg red_g(x,y) \lor \neg blue_g(x,y) \lor \neg green_g(x,y))
\end{array}
\]
More generally, we define a polynomial sized quantifier free formula
$INC_{g,c}(x,y)$ for any given constant $c$ with $2^{r-1} <c \leq 2^r$ using
propositions $p_1, \ldots, p_r$ and three colour predicates.
The formula asserts that $\#a(x,y)+ c = 2^r$.
Using this we can encode the constraint $\#a(x,y)=2^r-c$ for any $c$.
Similarly to $EQ_g(x,y)$, we can also
define formulas $LT_g(x,y)$ to denote that its counter $C_g(x) < C_g(y)$,
$GT_g(x,y)$ to denote that $C_g(x) > C_g(y)$, etc.
Hence any form of threshold counting
relation can be replaced by an equisatisfiable formula, with
all these global counters running from the beginning of the word to the end.
Thus we have a polynomially sized equisatisfiable reduction from
\fotwothr to \fotwobet.
\end{proof}

To complete the proof of Theorem~\ref{thm.fo2sat},
the upper bound for \fotwothr comes from
Lemma~\ref{lem.reduction}, an exponential translation
from \fotwobet to \binvtl using an order type argument similar to~\cite{EVW}
(our Theorem~\ref{thm.tlequivalence} also points to this equivalence),
and the \pspace upper bound for \binvtl (Theorem~\ref{thm.utlsat}).
Again using standard order type arguments~\cite{EVW}, we can translate
an $\fobetfac$ sentence to an exponentially sized \bfactl formula,
whose satisfiability is decidable in \pspace by Theorem~\ref{thm.utlsat}.


\section{Algebraic Characterizations}\label{AndreasHoward}

\subsection{Background on finite monoids and varieties}

For further background on the basic algebraic notions in this section, see Pin~\cite{Pin}.

A \textit{semigroup} is a set together with an associative multiplication.
It is a \textit{monoid} if it also has a multiplicative identity 1.

All of the languages defined by sentences of $FO[<]$ are regular languages.
The characterizations   for  languages definable  in this logic, as well as for the fragments $\fotwoless$ and $\fotwobet$,
are all given in terms of the of the \emph{syntactic monoid} $M(L)$ of $L$.
This is the transition monoid of the minimal deterministic
automaton recognizing $L$, and therefore finite.
Equivalently, $M(L)$ is  the smallest monoid $M$ that \textit{recognizes} $L$,
that is:  There is a homomorphism $h:A^*\to M$ and a subset $X\subseteq S$
such that $L=h^{-1}(X)$.

Let $M$ be a finite monoid. An \textit{idempotent} $e\in M$ is an element satisfying $e^2=e$.  If $m\in M$, then there is some $k\geq 1$ such that $m^k$ is idempotent.  This idempotent power of $m$ is unique, and we denote it by $m^{\omega}$.

A finite monoid is \textit{aperiodic} if it contains no nontrivial groups, equivalently, if it satisfies the identity $x\cdot x^{\omega}=x^{\omega}$ for all $x\in M$.  We denote the class of aperiodic finite monoids by ${\bf Ap}$.  ${\bf Ap}$ is a \textit{variety} of finite monoids: this means that it is closed under finite direct products, submonoids, and quotients.

A well-known theorem, an amalgam of results of McNaughton and Papert~\cite{MNP} and of Sch\"utzenberger~\cite{Sch1}, states that $L\subseteq A^*$ is definable in $FO[<]$ if and only if $M(L)\in{\bf Ap}$.  This situation is typical:  Under very general conditions, the languages definable in fragments of $FO[<]$ can be characterized as those whose syntactic monoids belong to a particular variety ${\bf V}$ of finite monoids. (See Straubing~\cite{Str-latin}.)

If $M$ is a finite monoid and $s_1,s_2\in M$, we write $s_1\leq_{\cJ}s_2$
if $s_1=rs_2t$ for some $r,t\in M$.  This is a preorder,
the so-called ${\cJ}$-ordering on $M$.
Let $E(M)$ denote the set of all idempotents in $M$.
If $e\in E(M)$, then $M_e$ denotes the submonoid of $M$ generated by the set
$\{m\in M: e\leq_{\cJ} m\}$.
The operation $M_e$ appears in an unpublished memo of Sch\"utzenberger
cited by Brzozowski~\cite{Brz}. He uses the submonoid generated by the
generators of an idempotent element $e$ of a semigroup.
For example, if $abc$ mapped to an idempotent element $e$,
$M_e$ would correspond to the language ${(a+b+c)}^*$.
The operation can be used at the level of semigroup and monoid classes.

Let {\bf V} be a variety of finite monoids.  We denote by ${\bf M}_e{\bf V}$ the family of all monoids $M$ such that for every $e\in E(M)$, the monoid $e M_e e$ belongs to {\bf V}.  The following property is easy to show.

\begin{prop}\label{prop.emee}
If {\bf V} is a variety of finite monoids, then so is ${\bf M}_e{\bf V}$.
\end{prop}

Let ${\bf I}$ denote the variety consisting of the trivial one-element monoid alone. We define the class
\[{\bf DA}={\bf M}_e{\bf I}.\]
That is, ${\bf DA}$ consists of those finite monoids $M$ for which $e M_e e = e$ for all idempotents $e\in M$. By Proposition~\ref{prop.emee}, ${\bf DA}$ is a variety of finite monoids. The variety ${\bf DA}$ was introduced by Sch\"utzenberger~\cite{Sch2} and it figures importantly in work on two-variable logic.   Th\'erien and Wilke showed that a language $L$ is definable in $FO^2[<]$ if and only if $M(L)\in{\bf DA}$~\cite{TW}.
Thus the variety $\meda$ has monoids $M$, all of whose submonoids
of the form $eM_e e$ for $e \in E(M)$, are in the variety \da.

\begin{exa}
Consider the language $L\subseteq {\{a,b\}}^*$ consisting of all words whose first and last letters are the same. The syntactic monoid of $L$ contains five elements
\[M(L)=\{1,(a,a),(a,b),(b,a),(b,b)\},\]
with multiplication given by $(c,d)(c',d')=(c,d')$, for all $c,d,c',d'\in\{a,b\}$.  Observe that every element of $M(L)$ is idempotent.  For every $e\neq 1$, $eM(L)e=e$, and if $e=1$, then ${M(L)}_e=1$.  Thus, $M(L)\in{\bf DA}$.  The logical characterization then tells us that $L$ is defined by a sentence of $FO^2[<]$.  Indeed, $L$ is defined by
\[\exists x(\forall y(x\leq y)\wedge\exists y(\forall x(x\leq y)\wedge(a(x)\leftrightarrow a(y)))).\]
\end{exa}

\begin{exa}
Consider the language ${(ab)}^*$.  We claimed earlier that it is not definable in $FO^2[<]$. We can prove this using the algebraic characterization of the logic.  The elements of the syntactic  monoid $M$ are
\[1,a,b,ab,ba,0.\]
The multiplication is determined by the rules $aba=a$, $bab=b$, and $a^2=b^2=0$. Then $ab$ and $ba$ are idempotents, and $M_{ab}=M_{ba} = M$. Thus $ab\cdot M_{ab}\cdot ab=\{ab,0\}$, which shows that $M\notin {\bf DA}$, and thus ${(ab)}^*$ is not definable in $FO^2[<]$.
\end{exa}

Certain classes of languages only admit such an algebraic characterization if we restrict to non-empty words, and in these instances, the characterization is typically given in terms of the syntactic semigroup $S(L)$ of the language and a variety of finite semigroups.  This is the case for $\fotwosuc$.  Th\'erien and Wilke~\cite{TW}   give an algebraic characterization of this class of languages as well: $L$ is definable in \fotwosuc if and only if for each idempotent $e\in S$, the monoid $N=eSe$ is in {\bf DA}.  (This is a general method for obtaining a variety of finite semigroups from a variety of finite monoids.) We apply this theorem in the example below.

\begin{exa}
Consider the language given by the regular expression ${(a+b)}^*bab^+ab{(a+b)}^*$.  We saw earlier that it is definable in \fotwobet, and claimed that it could not be defined in \fotwosuc.     For the language under discussion, let us denote the image of a word $w$ in the syntactic monoid by $\overline{w}$.  Then $e=\overline{b}$ is idempotent, and $f=\overline{baab}$ is an idempotent in $N=eSe$. Let $s=\overline{bab}$.  Then $s\in eSe$, and $fsf=f$, so $s\in M_f$.
We now have
$fsf=f\neq fssf$, since $fssf=\overline{babab}$ is the zero of $N$. Thus $f M_f f$ contains more than one element, so $eSe\notin {\bf DA}$. Consequently this language, while definable in \fotwobet, cannot be defined in \fotwosuc.
\end{exa}

Our characterization of $\fotwobetfac$, in the final section of the paper, will also be given in terms of the syntactic semigroup.


\section{Characterization of the logic \texorpdfstring{$\fobet$}{FO2[<,BET]}}\label{sec:licsmain}

Over this section and the next, we present the
algebraic characterization of the logic $\fobet$.

\begin{thm}\label{thm.licsmain}\label{thm.cslmain}
Let $L\subseteq A^*$.
If $L$ is definable in \fotwobet if and only if $M(L)\in {\bf M}_e{\bf DA}$.
\end{thm}

We will prove one direction (necessity of the condition $M(L)\in {\bf M}_e{\bf DA}$) later in this section, and the other direction (sufficiency) in Section~\ref{sec:fac}.

We can prove from rather abstract principles that there is some variety ${\bf W}$ of finite monoids that characterizes definability in \fotwobet in this way.  Our theorem implies that ${\bf W}$ coincides with ${\bf M}_e{\bf DA}$. The theorem provides an effective method for determining whether a given regular language is definable in \fotwobet, since we can compute the multiplication table of the syntactic monoid, and check whether it belongs to ${\bf M}_e{\bf DA}$.

\begin{exa}
In our example above, where $L={(ab)}^*$, all the submonoids $e M_e e$ are either trivial, or are two-element monoids isomorphic to $\{0,1\}$, which is in {\bf DA}.  Thus ${(ab)}^*$ is definable in \fotwobet, as we saw earlier by construction of a defining formula.
\end{exa}

The following corollary to Theorem~\ref{thm.licsmain} answers our original question of whether $FO[<]$ has strictly more expressive power than \fotwobet.
To prove it, we need only calculate the syntactic monoid of the given language, and verify that it is in {\bf Ap} but not in ${\bf M}_e{\bf DA}$.

\begin{cor}\label{cor.strictinclusion} The language given by the regular expression
${(a{(ab)}^*b)}^*$
is definable in $FO[<]$ but not in \fotwobet.
\end{cor}

\begin{proof}[Proof of Corollary.]
Let
$L= {(a{(ab)}^*b)}^*$.
It is easy to construct a sentence of $FO[<]$ defining $L$. We can give a more algebraic proof by working directly with the minimal automaton of $L$ and showing that its transition monoid is aperiodic.  This automaton has three states $\{0,1,2\}$ along with a dead state.  For each state $i=0,1,2$ we define $i\cdot a=i+1$, $i\cdot b=i-1$, where this transition is understood to lead to the dead state if $i+1=3$ or $i-1=-1$. The transition monoid has a zero, which is the transition mapping every state to the dead state. It is then easy to check that every $m\in M(L)$ is either idempotent, or has $m^3=0$, and thus $M(L)$ is aperiodic.
We now show that $L$ cannot be expressed in \fotwobet.  By Theorem~\ref{thm.licsmain}, if  $L$ is expressible then $M(L)\in {\bf M}_e{\bf DA}$. As before, we will denote the image of a word $w$ in $M(L)$ by $\overline {w}$. Easily, $e=\overline{ab}$ is an idempotent in $M(L)$, and both $m=\overline{(ab)a(ab)}$ and $f=\overline{(ab)a(ab)b(ab)}$ are elements of $e\cdot {M(L)}_e\cdot e$ with $f$ idempotent and $f\leq_{\cJ} m$ in $e\cdot {M(L)}_e\cdot e$.  So if $e\cdot {M(L)}_e\cdot e\in{\bf DA}$ we would have
$fmf=f$.  Now $f$ is the transition that maps state 0 to 0 and all other states to the dead state, and $m$ is the transition that maps 0 to 1 and all other states to the dead state.  Thus $fmf=0\neq f$, so $M(L)\notin {\bf M}_e{\bf DA}$.
\end{proof}

\subsection{Proof of necessity in Theorem~\ref{thm.licsmain}}

Here we prove the direction of Theorem~\ref{thm.licsmain} stating that every language definable in \fotwobet has its syntactic monoid in ${\bf M}_e{\bf DA}$.  This is all we needed to prove Corollary~\ref{cor.strictinclusion}.

First, we use a game argument to prove the following fact:

\begin{lem}\label{lem.morphismclosure}
Let $k\geq 0$, $A,B$ finite alphabets, and $f:B^*\to A^*$ a monoid homomorphism.  Let $w_1,w_2\in B^*$.  If  $w_1\equiv_k w_2$, then $f(w_1)\equiv_k f(w_2)$.
\end{lem}

\begin{proof}  We do this by a strategy-copying argument.  If $w_1\equiv_k,w_2$, then Player 2 has a winning strategy in the $k$-round game in $(w_1,w_2)$. We use it to produce a winning strategy for the $k$-round game in $(f(w_1),f(w_2))$.  We write
\[w_1=a_{11}\cdots a_{1r}, w_2=a_{21}\cdots a_{2s},\]
and
\[f(w_1)=v_{11}\cdots v_{1r}, f(w_2)=v_{21}\cdots v_{2s},\]
where $v_{ij}=f(a_{ij})\in B^*$.

  Player 2 will keep track of a `shadow game' in $(w_1,w_2)$ to determine her moves.
Let us assume that Player 1's move is in $f(w_1)$; the description is the same if we assume the move is in $f(w_2)$. If Player 1 places the pebble initially on a position of $f(w_1)$, or moves to such a position, then the position falls within some factor $v_{1i}=f(a_{1i})$.  Player 2, acting as Player 1, will place or move the pebble in the shadow game so that it is on the letter $a_{1i}$.  She then uses her strategy in the shadow game to  determine the reply, moving to $a_{2j}$ in $w_2$.  Since this is a legal reply for Player 2, we must have $a_{2j}=a_{1i}$.  Thus $v_{2j}=v_{1i}$.  Now suppose Player 1's move was on the $t^{th}$ letter of $v_{1i}$.  Then Player 2's move will be on the $t^{th}$ letter of $v_{2j}$.

It remains to show that in all cases this constitutes a legal reply for Player 2 in the game in $(f(w_1),f(w_2))$. Since $v_{1i}=v_{2j}$, the $t^{th}$ letters of the two words are the same, so the two pebbles wind up on identical letters.  To see that the two moves are in the same direction, observe that if Player 1 moves the pebble toward the right in $w_1$, then in the shadow game, Player 2 either moves to the right or keeps the pebble on the same position.  (This latter case is like a passed move, and Player 2 responds by not moving her pebble.) In either instance, the reply in $f(w_2)$ is a move toward the right, either within the factor $v_{2j}$, or from $v_{2j'}$ to $v_{2j}$ with $j'<j$.  To see that the sets of letters jumped over is the same, suppose that the move in $w_1$ was from the $t^{th}$ letter of $v_{1i'}$ to the $u^{th}$ letter of $v_{1i}$ with $i\leq i'$. If $i'<i$, the set of letters jumped over is the set of letters in $v_{1,i'+1},\ldots v_{1,i-1}$ together with the letters from positions $t$ to the end of $v_{1i'}$ and the letters from the start of $v_{1i}$ to position $u$, and similarly in $f(w_2)$. Equality of the two sets of letters follows from equality in the shadow game. If $i'=i$, then the set of jumped letters is just the set of letters strictly between positions $t$ and $u$ of $v_{1i}=v_{j2}$.

Observe that this argument works even if the morphism $f$ is erasing---that is, if $f(a)=1$ for some letters $a$.
\end{proof}

Now let $L\subseteq A^*$ be definable by a sentence of \fotwobet.   $L$ is a union of $\equiv_k$-classes for some $k$, and thus $L$ is recognized by the quotient monoid $N=A^*/\equiv_k$, so $M(L)$ is a homomorphic image of this monoid.  Consequently, it is sufficient to show that $N$ itself is a member of the variety ${\bf M}_e{\bf DA}$.  We denote by $\psi:A^*\to N$ the projection morphism onto this quotient.

Take $e=e^2\in N$ and $x,y\in e N_e e$. We will show
\[{(xy)}^{\omega}x{(xy)}^{\omega}={(xy)}^{\omega}.\]
This identity characterizes the variety {\bf DA} (see, for example Diekert, {\it et al.}~\cite{DGK}), so this will prove $N\in {\bf M}_e{\bf DA}$, as required. We can write
\[x=em_1\cdots m_r e, y=em_1'\cdots m_s'e,\]
where each $m_i$ and $m_j'$ is $\leq_{\cJ}$-above $e$. Thus we have
\[e=n_{2i-1} m_i n_{2i}=n_{2r+2j-1}m_j'n_{2r+2j},\]
for some $n_1,\ldots n_{2r+2s}\in N$.

Now let $B$ be the alphabet
\[\{a_1,\ldots,a_r,b_1,\ldots,b_s,c_1,\ldots,c_{2r+2s}\},\]
  We define a homomorphism $\phi:B^*\to N$ by mapping each $a_i$ to $m_i$, $b_j$ to $m_j'$, and $c_k$ to $n_k$.  Since $\psi:A^*\to N$ is onto, we also have a homomorphism $f:B^*\to A^*$ satisfying $\psi\circ f=\phi$.
We define words $v,{\bf a},{\bf b}, X_{S,T}\in B^*$, where $S,T>0$, as follows:
\[v=\prod_{i=1}^r c_{2i-1} a_i c_{2i}\cdot\prod_{i=r+1}^s c_{2i-1}b_{i-r}c_{2i},\]
where
${\bf a}=a_1\cdots a_r$,  ${\bf b}=b_1\cdots b_r$, and
$X_{S,T}={(v^S{\bf a}v^S{\bf b}v^S)}^T$.

 Observe that
\[\phi(v)=e, \phi(X_{S,T})={(xy)}^T,\phi({\bf a})=m_1\cdots m_r,\phi({\bf b})=b_1\cdots b_s.\]
Since all languages recognized by $N$ are definable in $FO[<]$, $N$ is aperiodic, and thus for sufficiently large values of $T$, we have $\phi(X_{S,T})={(xy)}^{\omega}$.  Thus for sufficiently large $T$ and all $S$ we have
\[\phi(X_{S,T}X_{S,T})={(xy)}^{\omega}{(xy)}^{\omega}={(xy)}^{\omega},\]
\[\phi(X_{S,T}{\bf a}X_{S,T})={(xy)}^{\omega}x{(xy)}^{\omega}.\]

The crucial step is provided by the following lemma.

\begin{lem}\label{lemma.xst}
Let $k\geq 0$.  For $S,T$ sufficiently large,
\[X_{S,T}X_{S,T}\equiv_k X_{S,T}{\bf a}X_{S,T},\]
\end{lem}

\noindent
Assuming this for the moment,  by Lemma~\ref{lem.morphismclosure} we have

\[f(X_{S,T}X_{S,T})\equiv_k f(X_{S,T}{\bf a}X_{S,T}).\]
So
\begin{eqnarray*}
{(xy)}^{\omega} &=&\phi(X_{S,T}X_{S,T})\\
&=& \psi(f(X_{S,T}X_{S,T}))\\
&=& \psi(f(X_{S,T}{\bf a}X_{S,T}))\\
&=& \phi(X_{S,T}{\bf a}X_{S,T})\\
&=& {(xy)}^{\omega}x{(xy)}^{\omega},
\end{eqnarray*}
as required.


\subsection{Proof of Lemma~\ref{lemma.xst}}

Let $r,B,{\bf a},{\bf b}$, and $v$ be as above.
 Let $R>0$. We will build words by concatenating the factors
\[{\bf a} , {\bf b}, v^R.\]

For example, with $R=4$, two such words are

\[v^4{\bf a}v^4{\bf b}v^8{\bf a}v^4{\bf a}v^{12}, v^4{\bf a}v^4{\bf ba}v^4{\bf a}v^{12}.\]
If the first and last factors of such a word are $v^R$, and if two consecutive factors always include at least one $v^R$, then we call it an \textit{$R$-word}.  The first word in the example above is a 4-word, but the second is not, because of the consecutive factors {\bf b} and {\bf a}. In what follows, we will concern ourselves exclusively with $R$-words. The way in which we have defined the word $v$ ensures that the factorization of an $R$-word in the required form is unique.
%

 Let $m\geq 0$, and $R> 2m$.  We will define special factors in $R$-words that we call \textit{$m$-neighborhoods}. One kind of $m$-neighborhood is a factor  of the form  $v^m{\bf a}v^m$ or $v^m{\bf b}v^m$, where the ${\bf a}$ or ${\bf b}$ is one of the original factors used to build the word.  In addition, we say that the prefix $v^m$ and suffix $v^m$ are also $m$-neighborhoods.  So, for example, the 1-neighborhoods in the 4-word in the example above are indicated here by underlining:

 \[\underline{v}\cdot v^2\cdot\underline{v{\bf a}v}\cdot v^2\cdot\underline{v{\bf b}v}\cdot v^6\cdot\underline{v{\bf a}v}\cdot v^2\cdot\underline{v{\bf a}v}\cdot v^{10}\cdot\underline{v}.\]

  The condition $R> 2m$ ensures that $m$-neighborhoods are never directly adjacent, so that every position belongs to at most one $m$-neighborhood, and some of the $v$ factors are contained in no $m$-neighborhood.

 Consider two marked words $(w_1, i_1), (w_2,i_2)$ where $w_1, w_2$ are $R$-words. We say these marked words are $\equiv_0^m$-equivalent if $w_1(i_1)=w_2(i_2)$, and if either $i_1$ and $i_2$ are in the same position in identical $m$-neighborhoods, or if neither $i_1$ nor $i_2$ belongs to an $m$-neighborhood.   For instance, if $m=2$, and $i_1$ is on the third position of a 2-neighborhood $v^2{\bf a}v^2$ in $w_1$, then $i_2$ will be on the third position of a 2-neighborhood $v^2{\bf a}v^2$ of $w_2$.  If $m=0$, then we only require $w_1(i_1)=w_2(i_2)$, so $\equiv_0^0$ equivalence is the same as $\equiv_0$-equivalence.

 We now play our game  in marked words $(w_1,i_1), (w_2,i_2)$,  where $w_1, w_2$ are $R$-words. We add the rule that at the end of every round, the two marked words $(w_1,j_1)$ and $(w_2,j_2)$ are $\equiv_0^m$-equivalent.  If Player 2 has a winning strategy in the $k$-round game with this additional rule, we say that the starting words $(w_1,i_1)$ and $(w_2,i_2)$  are  $\equiv_k^m$-equivalent. Once again, the case $m=0$ corresponds to ordinary $\equiv_k$-equivalence.

 We will call this stricter version of the game the \textit{$m$-enhanced game}. As with the original game, we can define a version of the $m$-enhanced game for ordinary (that is, unmarked) $R$- words: In the first round, Player 1 places his pebble on a position in either of the words, and Player 2 responds so that the resulting marked words are $\equiv_0^m$-equivalent. Play then proceeds as described above for $k-1$ additional rounds.  We write $w_1\equiv^m_k w_2$ if Player 2 has a winning strategy in this $k$-round $m$-enhanced game.

 We claim that for each $m\geq 0$, $k\geq 1$, there exists $R$ such that if $S,T\geq R$,
 \[X_{S,T}X_{S,T}\equiv_k^m X_{S,T}{\bf a}X_{S,T}.\]

 The case $m=0$ is the desired result.

  We prove this by induction on $k$, first considering the case $k=1$ with arbitrary $m\geq 0$. Choose $R> 2m$, and $S,T\geq R$. If Player 1 plays in either word inside one of the factors $X_{S,T}$, then Player 2 might try to simply mimic this move in the corresponding factor $X_{S,T}$ in the other word.  This works unless Player 1 moves near the center of one of the two words.  For example, if Player 1 moves in the final position of the first $X_{S,T}$ in $X_{S,T}X_{S,T}$, then this position is contained inside a factor $v$ and does not belong to any $m$-neighborhood, but the corresponding position in the other word belongs to a neighborhood of the form $v^m{\bf a}v^m$, so the response is illegal. Of course we can solve this problem: $X_{S,T}$ itself contains factors $v$ that do not belong to any $m$-neighborhood (because $R> 2m$).  Conversely, if Player 1 moves anywhere in the central $m$-neighborhood $v^m{\bf a}v^m$ in $X_{S,T}{\bf a}X_{S,T}$ then Player 2 can find an identical neighborhood inside $X_{S,T}$ and reply there.

Now let  $k\geq 1$ and suppose that the Proposition is true for this fixed $k$ and all $m\geq 0$. We will show the same holds for $k+1$. Let $m\geq 0$. Then by the inductive hypothesis, there exist $S,T$ such that
\[X_{S,T}X_{S,T}\equiv^{m+1}_k X_{S,T}{\bf a}X_{S,T}.\]
We will establish the proposition by showing
\[X_{S,T+1}X_{S,T+1}\equiv^{m}_{k+1}X_{S,T+1}{\bf a}X_{S,T+1}.\]
Observe that
\[X_{S,T+1}X_{S,T+1}=X_{S,1}(X_{S,T}X_{S,T})X_{S,1},\]
\[X_{S,T+1}{\bf a}X_{S,T+1}=X_{S,1}(X_{S,T}{\bf a}X_{S,T})X_{S,1}.\]
We will call the factors $X_{S,1}$ occurring in these two words the \textit{peripheral factors}; the remaining letters make up the \textit{central regions}. We prove the proposition by presenting a winning strategy for Player 2 in the $(k+1)$-round $m$-enhanced game in these two words.

For the first $k$ rounds, Player 2's strategy is as follows:
\begin{itemize}
\item If Player 1 moves into either of the peripheral factors in either of the words, Player 2 responds at the corresponding position of the corresponding peripheral factor in the other word.
\item If Player 1 moves from a peripheral factor into the central region of one word, Player 2 treats this as the opening move in the $k$-round $(m+1)$-enhanced game in the central regions, and responds according to her winning strategy in this game.
\item If Player 1 moves from one position in the central region of a word to another position in the central region, Player 2 again responds according to her winning strategy in the $k$-round $(m+1)$-enhanced game in the central regions.
\end{itemize}

%
%

\noindent
 We need to show that each move in this strategy is actually a legal move in the game---that in each case the sets of letters jumped by the two players are the same, and that the $m$-neighborhoods match up correctly.  It is trivial that the $m$-neighborhoods match up---in fact, the $(m+1)$-neighborhoods do, so we concentrate on  showing that the sets of jumped letters are the same.

This is clearly true for moves that remain within a single peripheral factor or move from one peripheral factor to another.  What about the situation where Player 1 moves from a peripheral factor to the central region?  We can suppose without loss of generality that the move begins in the $i^{th}$ position of the left peripheral factor in one word, and jumps right to the $j^{th}$ position of the central region of the the word. If $j\leq |v|$, then this move is within an $m$-neighborhood of the central factor, namely the prefix of the central factor of the form $v^m$, so Player's 2 reply will be at the identical position in the corresponding neighborhood, and thus at position $j$ of the central factor in the other word.  Thus the two moves jumped over precisely the same set of letters. If $j>|v|$, then Player 1's move jumps over all the letters of $v$.  Since Player 2's response, owing to the condition on neighborhoods, cannot be within the first $|v|$ letters of the central factor of the other word, her move too must jump over all the letters of $v$. The same argument applies to moves from the central region into a peripheral factor.

This shows that Player 2's strategy is successful for the first $k$ rounds, but we must also show that Player 2 can extend the winning strategy for one additional round. If after the first $k$ rounds, the two pebbles are in peripheral factors, then Player 1's next move is either within a single peripheral factor, from one peripheral factor to another, or from a peripheral factor into the central region.  In all these cases Player 2 responds exactly as she would have during the first $k$ rounds. As we argued above, this response is a legal move.

If Player 1 moves from the central region to one of the peripheral factors, Player 2 will move to the same position in the corresponding peripheral factor in the other word.  Again, we argue as above that this is a legal move.

The crucial case is when Player 1's move is entirely within one of the central factors. Now we can no longer use the winning strategy in the game in the central factors to determine Player 2's response, because she has run out of moves. We may suppose, without loss of generality, that Player 1's move is toward the left. We consider whether or not the set of letters that Player 1 jumps consists of all the letters in $B$.  If it does, then Player 2 can just locate a matching neighborhood somewhere in the left peripheral factor in the other word and move there.  In the process, Player 2 also jumps to the left over all the letters of $B$.

What if Player 1 jumps over a proper subset of the letters of $B$?   Player 2 responds by moving the same distance in the same direction in the other word. Why does this work? After $k$ rounds, the pebbles in the two words must be on the same letters and in the same positions within matching $(m+1)$-neighborhoods, or outside of any $(m+1)$-neighborhood.  Thus we cannot have the situation where, for example, the pebble in one word is on a letter $b_i$ within a factor ${\bf b}$, and in the other word on the same letter $b_i$ within a factor $v$.  If Player 1 moves the pebble from a factor $v$ into a factor ${\bf b}$ or ${\bf a}$---let us say ${\bf b}$--- then the original position of the pebble was within a 1-neighborhood.  Therefore Player 2's pebble was in the corresponding position in a matching 1-neighborhood, so the response moves this pebble to ${\bf b}$ as well, and thus jumps over the same set of letters. (This is the only potentially problematic case.)  How can we ensure that the $m$-neighborhoods match up correctly after this move?  If  the move takes Player 1's pebble into an $m$-neighborhood, then the move stayed within a single $(m+1)$-neighborhood.  This means that the original position of Player 2's pebble was in the corresponding position of an identical $(m+1)$-neighborhood, so that after Player 2's response, the two pebbles are in matching $m$-neighborhoods.



\section{The factorization sequence}\label{sec:fac}

Sufficiency in Theorem~\ref{thm.cslmain} will be proved by induction on the alphabet size.
The proof is somewhat intricate.
We develop new techniques of factorization which are amenable to
simulation using logic. The bulk of the proof is combinatorics on words and finite model theory.
At the end we will apply the algebraic characterization of $\fotwosuc$ from~\cite{TW} cited in the previous section.

Throughout this section, we will suppose that $h:A^*\to M$ is a homomorphism onto a finite monoid $M$  that satisfies $e M_e e\in{\bf DA}$ for all idempotents $e$. Sufficiency in Theorem~\ref{thm.cslmain} is equivalent to showing that the language  $h^{-1}(m)$ is definable in $FO^2[<,bet]$ for all $m\in M$.

The proof is trivial for a one-letter alphabet, so assume $|A|>1$
and that the theorem holds for all strictly smaller alphabets.

For now we distinguish a letter $a\in A$, and restrict our attention to
\emph{$a$-words} $w$ with the following three properties:
\begin{itemize}
\item $\alpha(w)=A$
\item $a$ is the first letter of $w$
\item $a$ is the \textit{last} letter to appear in a \textit{right-to-left} scan of $w$; that is, $w=xay$ where $\alpha(y)=A\backslash\{a\}$.
\end{itemize}

\noindent
We describe an algorithm for constructing a sequence of factorizations for any $a$-word.  Each step of the algorithm is divided into two sub-steps, and we will refer to each of these sub-steps as a factorization scheme.  The factors that occur in each scheme are formed by concatenating factors from the previous scheme. That is, at each step, we clump smaller factors into larger ones, so the number of factors decreases (non-strictly) at each step.

We begin by putting a linear ordering $<$ on the set of proper subalphabets of $A$  that contain the letter $a$. This will be a topological sort of the subset partial order.  That is, if $B,C$ are two such subalphabets with $B\subsetneq C$, then $B<C$, but otherwise the ordering is arbitrary.  For example, with $A=\{a,b,c,d\}$, we can take
\[\{a\} < \{a,b\}<\{a,c\}<\{a,b,c\}<\{a,d\}<\{a,b,d\}<\{a,c,d\},\]
as one of many possibilities.
One way to think about our techniques is as a refinement of Th\'erien and
Wilke's combinatorial characterization of \da~\cite{TW}
which only used the inclusion order over an alphabet.

Here is the algorithm, which is the new development over \da:

\begin{itemize}
\item Initially factor $w$ as
$au_1\cdots au_k$,
where each $\alpha(u_i)$ is properly contained in $A$.
\item For each  proper subalphabet  $B$ of $A$ with $a\in B$, following the linear order
\begin{itemize}
	\item For each  factor $u$ such that $\alpha(u)=B$, combine all sequences of consecutive factors of this kind into a single factor. We say that $B$ is now \textit{collected}.
	\item For each   factor $u$ such that $\alpha(u)=B$,   combine each such factor with the factor immediately to its right.  We say that $B$ is now \textit{capped}.
\end{itemize}

\end{itemize}

\noindent
Here is an example.  We begin with an $a$-word and its initial factorization:
\begin{quote}
$adccdcc\cdot adc\cdot a\cdot a\cdot a\cdot addccdcccdbcdc\cdot a\cdot ac\cdot abcbbd$
\end{quote}

We use the ordering in the example above. We start with $B=\{a\}$ and collect $B$:

\begin{quote}
$adccdcc\cdot adc\cdot {\rm aaa}\cdot addccdcccdbcdc\cdot {\rm a}\cdot ac\cdot abcbbd$
\end{quote}
then cap it:

\begin{quote}
$adccdcc\cdot adc\cdot {\rm aaaaddccdcccdbcdc}\cdot {\rm aac}\cdot abcbbd$\end{quote}

We choose $B=\{a,b\}$.  There is nothing to do here, because no factor contains just $a$ and $b$. $B=\{a,c\}$ is already collected, because there is no pair of consecutive factors with this alphabet, so we cap it:

\begin{quote}
$adccdcc\cdot adc\cdot {\rm aaaaddccdcccdbcdc}\cdot {\rm aacabcbbd}$\end{quote}

The next subalphabet in order that occurs as a factor is $\{a,c,d\}$.  We collect:

\begin{quote}
${\rm adccdccadc}\cdot {\rm aaaaddccdcccdbcdc}\cdot {\rm aacabcbbd}$
\end{quote}
then cap:
\begin{quote}
${\rm adccdccaddadaaaaddccdcccdbcdc\cdot acabcbbd}$
\end{quote}

Let us make a few general observations about this algorithm:  Every proper subset of $A$ containing $a$ that occurs as the alphabet of a factor will eventually be capped, because the rightmost factor $au_k$ of the initial factorization contains all the letters of $A$.  Once $B$ has been collected, there is no pair of consecutive factors with content $B$. Once $B$ has been capped, there are no more factors with content $B$ nor with strictly smaller content. Thus at the end of the process, every factor contains all the letters of $A$.

Note as well that immediately after a subalphabet $B$ is collected to create a (possibly) larger factor $u$, both the factor immediately to the right of $u$ and immediately to the left of $u$ must contain a letter that is not in $u$.


\subsection{Starts and jumps}

We establish below several model-theoretic properties  of the factorization schemes produced by the above algorithm.

\begin{lem}\label{lem.factschemes}
\begin{enumerate}[(a)]
\item\label{item.start}
There is a formula \emph{start} in $\fobet$ such that for all $a$-words $w$,
$(w,i)\models start(x)$ if and only if $i$ is the first position in a factor of $w$.

\item\label{item.next}
Let $\phi_1(x)$ be a formula in $\fobet$.
Then there is a formula \emph{next} in $\fobet$ with the following property: Let $w$ be an $a$-word and let $i$ be the first position in some factor of $w$  that is not the rightmost factor. Then $(w,i)\models\phi_1(x)$ if and only if $(w,i_{succ})\models next(x)$, where $i_{succ}$ is the first position in the \textit{next} factor of $w$.
We also define the analogous property, with `rightmost' replaced by `leftmost',
\emph{next} by \emph{previous}, and $i_{succ}$ by $i_{pred}$.
\end{enumerate}
\end{lem}

\begin{proof}
We prove these properties by induction on the construction of the sequence of
factorization schemes.  That is, we prove that
they hold for the initial factorization scheme,
and that they are preserved in each sub-step of the algorithm.
For the induction, we will use Ehrenfeucht-Fra\"{\i}ss\'e games
for the logic $\fobet$
to argue for the existence of the formula $start=start_{\tau}$
(see Section~\ref{sec:efgame}).

We note that the claim in Item~\ref{item.next} implies
the condition on games (possibly with different parameters).
If for every formula $\phi_1$ there is a corresponding `successor'
formula $next$, then there is some constant $c$ such that
${\tt qd}(next)\leq c+{\tt qd}(\phi_1)$,
where {\tt qd} denotes quantifier depth.
Suppose that Player 2 wins the $(k+c)$-round game in
$(w,i_{succ}),(v,i_{succ})$. Then  $(w,i_{succ})\equiv_{k+c}(v,j_{succ})$.
Consider the formula $\phi_1$ that defines the $\equiv_k$-class of $(w,i)$.
Then $(w,i_{succ})\models next$, so $(v,j_{succ})\models next$.
Thus $(v,j)\models\phi_1$, so $(w,i)\equiv_k(v,j)$,
and Player 2 wins the $k$-round game in these words.

\bigskip

We begin with Item~\ref{item.start}:  For the initial factorization, we simply take $start(x)$ to be $a(x)$:  the factor starts are exactly the positions that contain $a$. We now assume that $\tau$ is some factorization scheme in the sequence, and that for the preceding factorization scheme $\sigma$,  the required formula, which we denote $start_{\sigma}$, exists.

To establish this formulation, let $(w,i), (w',j)$ be as described.  Since, by the inductive hypothesis, the formula  $start_{\sigma}$ for the preceding scheme $\sigma$ exists, we can treat this as if it is an atomic formula, in describing our game strategy. Observe that $i$ must also be the start of a factor of $w$ according to the previous factorization scheme  $\sigma$.  We write this as $start_{\sigma}(i)$ rather than the more verbose $(w,i)\models start_{\sigma}(x)$. If $j$ does not satisfy $start_{\sigma}(j)$,  then by induction we are done, and can take the number $k$ of rounds to be the quantifier depth of $start_{\sigma}$.  Thus $j$ is the start of a factor with respect to  the scheme $\sigma$,  not  with respect to $\tau$.
This can happen in one of two ways, depending on whether the most recent sub-step collected a subalphabet $B$, or capped a subalphabet $B$.

In the first case, we will describe a winning strategy for Player 1 in a game that lasts just a few more rounds than the game for the previous scheme.   Position $j$ was the start of a factor in the prior scheme $\sigma$, and has been collected into a larger factor that begins at position to the left of $j$. First suppose that $i$ is the start of a factor with content different from $B$.  Then this factor must contain some $c\notin B$. Player 1 then wins as follows:  He moves right  in $(w',j)$, jumping to the start $j'$ of the next factor (which must satisfy $start_{\sigma}(j')$).  In so doing, all the letters he jumps belong to $B$.  Player 2 must also jump to the right in $(w,i)$, and must also land on the start of a factor in the scheme $\sigma$; otherwise, by induction, Player 1 will win the game in the next $k$ rounds.  But to do so, Player 2 will have to jump over a position containing $c$, so she cannot legally make this move. Thus $i$ must be the start of a factor with content $B$.  In this case, Player 1 moves left in $(w',j)$  to $j''$, the start of the previous factor with respect to $\sigma$. In doing so, he jumps over letters in $B$.  Now Player 2 must also jump to the left  in $(w,i)$ to a position that was the start of a factor with respect to $\sigma$, but must jump over a letter not in $B$ to do this, so Player 1 wins again. (See Figure~\ref{fig:game1}.)

\begin{figure}
\begin{center}
\includegraphics[width=4in,clip=true]{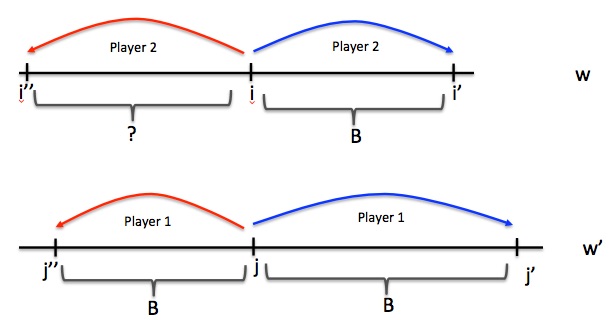}
\end{center}
\caption{Game-based proof of definability of factor starts. The figure shows the two words just after the step collecting the subalphabet $B$.  We suppose $i,j$ are factor starts for the preceding factorization scheme $\sigma$, and that $i$, but not $j$, is a factor start for the present scheme $\tau$. This means that the factor with respect to $\sigma$ beginning at $j$ was joined to the previous factor as a result of the collection. If Player 1 moves to the start $j'$ of the next factor of $w'$ with respect to $\sigma$ (blue arrow), then he jumps over precisely the letters of $B$. Thus for Player 2 to have a response, $i$ must be the start of a factor with alphabet $B$. But this means that the factor with respect to $\sigma$ in $w$ that precedes $i$ must contain a letter not in $B$.  As a result, Player 2 cannot reply to a  move by Player 1 to the start $j''$ of the factor with respect to $\sigma$ that precedes $j$ (red arrow).}%
\label{fig:game1}
\end{figure}

In the second case, where $B$ was capped, $j$ was the start of a factor that immediately followed a newly-collected factor with content $B$.  Player 1 jumps left to $j'$, the start position of this factor, and in doing so jumps over a segment with  content $ B$.  Thus Player 2 must jump to the start of a factor with respect to $\sigma$. For this to be a legal move, the segment she jumps must have content $B$. However, this is impossible, for any factor with this content in the scheme $\sigma$ would have been capped by the following factor, so that $i$ cannot be the start of a factor for $\tau$. (Figure~\ref{fig:game2}.)

\begin{figure}
\begin{center}
\includegraphics[width=4in,clip=true]{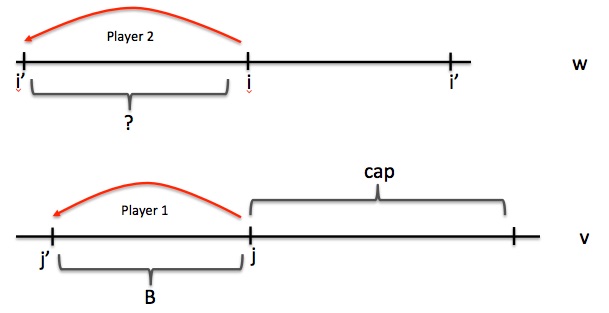}
\end{center}
\caption{This shows the case just after the step that caps the subalphabet $B$.  Again suppose $i,j$ are factor starts for the preceding factorization scheme $\sigma$, and that $i$, but not $j$, is a factor start for the present scheme $\tau$. If Player 1 moves in $v$ from $j$ to $j'$, the start of the factor preceding $j$ with respect to $\sigma$, then only letters in $B$ are jumped. If Player 2 moves left from $i$ to another factor start with respect to $\sigma$, she will have to jump over letters that are not in $B$, because all factors with alphabet $B$ have been capped; thus Player 2 cannot respond to this move.}%
\label{fig:game2}
\end{figure}

Now for Item~\ref{item.next}.  Again, we use a game argument.  We claim it will be enough to establish the following for sufficiently large values of $k$:  Let $(w,i),(v,j)$ be marked words, where $i,j$ are the starts of factors, and let $(w,i_{succ}),(v,i_{succ})$ be the same words, where the indices $i_{succ},j_{succ}$ mark the start of the successor factors. If Player 1 has a winning strategy in the $k$-round game in $(w,i),(v,j)$, then he has a winning strategy in the $k'$-round game in $(w,i_{succ}),(v,i_{succ})$ for some $k'$ that depends only on $k$ and the alphabet size, and not on $v$ and $w$.  Equivalently, if Player 2 wins in $(w,i_{succ}),(v,i_{succ})$ then she wins in $(w,i),(v,j)$.   Of course, there is the analogous formulation for \emph{previous}.

\bigskip

So we will suppose Player 1 has a winning strategy in the $k$-round game in $(w,i),(v,j)$, where $k$ is at least as large as the quantifier depth of $start_{\tau}$.  We will prove the existence of a strategy in $(w,i_{succ}),(v,j_{succ})$  for the $k'$-round game, where $k'$ is larger than $k$.  (By tracing through the various cases of the proof carefully, you can figure out how large $k'$ needs to be.)  What we will show in fact is that for each $\tau$, Player 1 can force the starting configuration $(w,i_{succ}),(v,j_{succ})$ to the configuration $(w,i),(v,j)$, and from there apply his winning strategy in $(w,i),(v,j)$.

 The base step is where $\tau$ is the initial factorization scheme.  Here the factor starts are just the positions where the letter $a$ occurs. Player 1 begins by jumping from $i_{succ}$ to $i$.  For Player 2 to respond correctly, she must jump from $j_{succ}$ to $j$, because she is required to move left and  land on a position containing $a$ while jumping over a segment that does not contain the letter $a$.

 So now we will suppose that $\tau$ is not the initial factorization scheme. We again denote the previous factorization scheme by $\sigma$.  We assume that the property in Item~(\ref{item.next}) holds for $\sigma$. Thanks to what we proved above, we know that the property in Item~(\ref{item.start}) holds for both $\tau$ and $\sigma$.   This means that we can treat $start_{\tau}$ and $start_{\sigma}$ essentially as atomic formulas.

If $i_{succ}$ is also the successor of $i$ (that is, the start of the next factor) with respect to the previous factorization scheme $\sigma$, and $j_{succ}$ is the successor of $j$, then we have the desired result by induction.  Thus we may suppose that one or both of the factor starts, either between $i$ and $i_{succ}$ or between $j$ and $j_{succ}$, or both, were eliminated in the most recent sub-step of the algorithm.


Let us suppose first that the most recent sub-step was a collection step, collecting the subalphabet $B$.   Player 1 jumps from $i_{succ}$ left to $i$.  The set of jumped letters is $B$. Player 2 must respond by jumping to some $j'<j_{succ}$ where $j'$ satisfies $start_{\tau}$.  If $j'<j$, then the set of jumped letters necessarily contains a letter not in $B$, so such a move is not legal.  Thus $j'=j$. Player 1 now follows his winning strategy in $(w,i),(v,j)$.  The identical strategy works for the predecessor version, because any factor following the sequence of collected factors must contain a letter not in $B$.

So suppose that the most recent sub-step was a capping step, and that the subalphabet $B$ was capped.  We may suppose that there is some $i'$ with $i<i'<i_{succ}$ such that $start_{\sigma}(i')$, but  not $start_{\tau}(i')$. Thus the interval from $i$ to $i'-1$ has content $B$ and constitutes a factor that was collected during the prior sub-step, before being capped in the present one. Player 1 uses his strategy from the previous factorization scheme to force the configuration to $(w,i'),(v,j')$, where $j'$ is the start of the factor preceding $j$ in the scheme $\sigma$.  Observe that we must have that $j'$ does not satisfy $start_{\tau}$ because $i'$ does not satisfy $start_{\tau}$.  Thus $j<j'<j_{succ}$, so the interval from $j$ to $j'-1$ is also a factor with content $B$ that was collected during the previous substep.  Player 1 now moves from $i'$ left to $i$.  Player 2 must respond with a move to $j''\leq j$ such that $start_{\tau}(j'')$ holds.  We cannot have $j''<j$, for then the set of jumped letters would include a letter not in $B$.  Thus $j''=j$, and the game is now in the configuration $(w,i),(v,j)$.

The strategy for a capped step in the predecessor game uses the same idea:  We may assume there is some $i'$ with $i_{prec}<i'<i$ such that the interval from $i_{prec}$ to $i'-1$  has content $B$ and constitutes a factor that was collected during the prior sub-step, before being capped in the present one.  Thus in the previous scheme $\sigma$, $i'$ was the successor position of $i_{prec}$.  Player 1 uses his strategy from the previous scheme to force the game to the configuration $(w,i'),(v,j')$, where $j'$ is the successor of $j_{prec}$ in the scheme $\sigma$.  We must have the set of jumped letters to be $B$ in each case, so the intervals from $i'$ to $i-1$ and $j'$ to $j-1$ are the caps applied in the scheme $\tau$, and thus $i$ is the successor of $i'$, and $j$ the successor of $j'$, in the scheme $\sigma$.  Player 1 now uses his strategy for the scheme $\sigma$ to force the game from the configuration $(w,i'),(v,j')$ to $(w,i),(v,j)$.
\end{proof}


\subsection{Simulating factorization in logic}\label{sec:sim}

A factorization scheme $\sigma$ gives a factorization $\sigma(w)=(w_1,\ldots,w_k)$ of an $a$-word $w$.  This in turn gives a word $\sigma_h(w)=m_1\cdots m_k\in M^+$.  We say that $\sigma$ \textit{admits simulations} if the following properties hold.

\begin{itemize}

\item For each sentence $\psi\in \fotwosuc$ over the alphabet $M$, there exists a sentence $\phi\in \fobet$ over the alphabet $A$ with the following property.  Let $w$ be an $a$-word.
\[w\models \phi\quad{\tt iff}\quad {\sigma}_h(w)\models\psi.\]

\item For each formula $\psi(x)\in \fotwosuc$ with one free variable over the alphabet $M$, there exists a formula $\phi(x)\in \fobet$ with one free variable over the alphabet $A$ with the following property. Let $w$ be an $a$-word, $1\leq i\leq k$ and let $j_i$ be the position within $w$ of the first letter of $w_i$ in $\sigma(w)$.  Then

\[(w,j_i)\models \phi(x)\quad{\tt iff}\quad ({\sigma}_h(w),i)\models\psi(x).\]

\end{itemize}

\begin{lem}[Simulation]%
\label{lem.simulations} Each factorization scheme in our sequence admits simulations.
\end{lem}

It is useful to have abbreviations for
commonly used subformulas of $\fobet$.  If $B$ is a subalphabet of $A$,
we write $\boxm{B}(x,y)$ to mean the conjunction of $\neg c(x,y)$ over
all $c\notin B$; in other words, `every letter between $x$ and $y$ belongs to $B$'.
$\boxm{a}(x,y)$ is always true if $y\leq x$ because
$b(x,y)$ is false for every $b\in A$ whenever $y\leq x$.
We denote by $\exact{B}(x,y)$ the conjunction of $\boxm{B}(x,y)$
together with the conjunction of $b(x,y)$ over all $b\in B$;
in other words, $B$ is exactly the set of letters between $x$ and $y$.

\begin{proof} The first claim in the Theorem follows easily from the second.  So we will begin with the formula $\psi(x)\in \fotwosuc$ over $M$ and and show how to produce $\phi(x)$.  We prove this by induction on the construction of formula $\psi$.  So the base case is where $\psi(x)$ is an atomic formula $m(x)$, where $m\in M$. This means that for each factorization scheme $\sigma$, we have to produce a formula $\phi_{m,\sigma}(x)$ such that for an $a$-word $w$,
$(w,i)\models\phi_{m,\sigma}(x)$ if and only if the factor starting at $i$ maps to $m$ under $h$.

We do this by induction on the sequence of factorization schemes.  In the initial factorization, every factor is of the form $au$, where $a\notin \alpha(u)$.  This factor maps to $m$ if and only if $h(u)=m'$ for some $m'\in M$ satisfying $h(a)\cdot m'=m$.  Since we suppose the main theorem holds for every alphabet strictly smaller than $A$, there is a sentence $\rho\in \fobet$ such that $u\models \rho$ if and only if $h(u)=m'$ where  $h(a)\cdot m'=m$.  We now relativize $\rho$ to obtain a formula $\rho'$ with one free variable that is satisfied by $(w,i)$ if and only if the factor of $w$ starting at $i$ has the form $au$, where $u\models\rho$.  To do this, we do a standard relativization trick, working from the outermost quantifier of $\rho$ inward.
We can assume that all the quantifiers at the outermost level quantify the variable $y$.  We replace each of these quantified formulas $\exists y\eta(y)$ by $\exists y(y>x\wedge \neg a(x,y)\wedge\eta(y))$.  Similarly, as we work inward, we rewrite each occurrence of $\exists z'(z'>z\wedge\eta)$ and $\exists z'(z'<z\wedge\eta)$, where $\{z,z'\}=\{x,y\}$, by adding the clause $\neg a(z,z')$ or $\neg a(z',z)$.  In essence, each time we jump left or right to a new position, we check that in so doing we did not jump over any occurrence of $a$, and thus remain inside the factor.

We now assume that $\tau$ is not the initial factorization scheme,
and that the formula $\phi_{m,\sigma}(x)$ exists for the preceding
factorization scheme $\sigma$.
We first consider the case where $\tau$ was produced during a step that
collected a subalphabet $B$.
Observe that we can determine within a formula whether $i$ is
the start of a factor that was produced during this collection step,
with the criterion
\[\exists y(x<y\wedge start_{\tau}(y)\wedge \exact{B}(x,y)).\]
(This includes the case where the collection is trivial
because there is only one factor to collect.)
If this condition does not hold,
then we can test whether the factor maps to $m$ with the formula produced
during the preceding step.
So we suppose that $i$ is the first position of one of the
new `collected' factors.
Since $B\subsetneq A$, there is a sentence $\rho$  of $\fobet$
satisfied by exactly the words over this smaller alphabet that map to $m$.
Once again, we must relativize $\rho$ to make sure that whenever
we introduce a new quantifier $\exists x(y>x\wedge\cdots)$ or
$\exists x(y<x\wedge\cdots)$ we do not jump to a position outside the factor.
To do this, we can replace $\exists x(y>x\wedge\cdots)$ by
\[\exists x(y>x\wedge\boxm{B}(x,y)\wedge\exists x(y<x\wedge start_{\tau}(x)\wedge\boxm{B}(x,y))).\]
In other words, we did not jump over any letter not in $\{a\}\cup B$,
and there is a factor start farther to the right that we can reach without
jumping over any letter not in $B$.
We replace $\exists x(y<x\wedge\cdots)$ by
\[\exists x(y<x\wedge\boxm{B}(y,x)\wedge\exists x(x\leq y\wedge start_{\tau}(x)\wedge\boxm{B}(x,y))),\]
using essentially the same idea.

Now suppose that $\tau$ was produced during a step that
capped the subalphabet $B$.
Again, we can write a formula that says that $i$ is the start of
a new factor produced in this process:
it is exactly the formula that said $i$ was the start of a factor that
collected $B$ in the preceding scheme $\sigma$.
So we only need to produce a formula that says the factor of $w$
beginning at $i$ maps to $m$ under the assumption that this is
one of the new `capped' factors.
Our factor has the form $u_1u_2$, where $u_2$ is the cap
and $u_1$ is the factor in which $B$ was collected.
We consider all pairs $m_1,m_2$ such that $m_1\cdot m_2=m$.
We know that there are formulas $\rho_1(x)$ and $\rho_2(x)$ telling us
that the factors  in the preceding scheme $\sigma$ map to $m_1$ and $m_2$.
We use the same formula $\rho_1(x)$, and take its conjunction with $next(x)$,
the successor formula derived from $\rho_2(x)$ by means of
Item~(\ref{item.next}) in Lemma~\ref{lem.factschemes}.
We are using the fact that the start of $u_2$ is the successor of
the start of $u_1$ under the preceding scheme $\sigma$.

We are almost done
(and we no longer need to induct on the sequence of factorization schemes)
because \fotwosuc formulas can be reduced to a few normal forms~\cite{EVW}.
Let us first suppose that our formula $\psi$ has the form
$\exists x(\suc(x,y)\wedge \kappa(x))$.
The inductive hypothesis is that there is a formula $\mu$ simulating $\kappa$.
Let $previous$ be the predecessor formula whose existence is given by
Item~(\ref{item.next}) of Lemma~\ref{lem.factschemes}.
We claim that $previous$ simulates $\psi$.  To see this, suppose
$w$ is an $a$-word, and $j_i$ is the position where the $i^{th}$ factor of $w$ begins.

Suppose $(w,j_i)\models previous$.
Then $(w,j_{i+1})\models\mu$.

So $(\sigma_h(w),i+1)\models\kappa$,
which gives $(\sigma_h(w),i)\models\psi$.

This implication also runs in reverse,
so we have shown that $previous$ simulates $\psi$.
Using the successor formula in place of the predecessor formula gives us
the analogous result for $\psi$ in the form $\exists x(\suc(y,x)\wedge \kappa(x))$.
\end{proof}


\subsection{Proof of sufficiency in Theorem~\ref{thm.cslmain}}\label{sec:charzn}

Again, we assume $|A|>1$ and that the theorem holds for all strictly smaller
alphabets.  Let $m\in M$, where $M$ satisfies the $\meda$ property.
	We need to show $h^{-1}(m)$ is defined by a sentence of $\fobet$.  As an overview, we will first, through a series of quite elementary steps, reduce this to the problem of showing that for each $a\in A$ and $s\in M$, the set of $a$-words mapping to $s$ is defined by a sentence of $\fobet$.  We then use Lemma~\ref{lem.simulations} on simulations, together with the algebraic characterization of $\fotwosuc$ cited earler, to find a defining sentence for the set of $a$-words that map to $s$.

First note that
$h^{-1}(m)=\bigcup_{B\subseteq A}\{w\in h^{-1}(m):\alpha(w)=B\}$.

It thus suffices to find, for each subalphabet $B$,
a sentence $\psi_B$ of $\fobet$ defining the set of words
$\{w\in h^{-1}(m):\alpha(w)=B\}$.
We then obtain a sentence for $h^{-1}(m)$ as
\[\bigvee_{B\subseteq A}(\psi_B\wedge\bigwedge_{b\in B}\exists xb(x)\wedge\bigwedge_{b\notin B}\neg\exists x b(x)).\]
Since we obtain the sentences $\psi_B$ for proper subalphabets $B$ of $A$
by the induction hypothesis, we only need to find $\psi_A$.

For each $w$ with $\alpha(w)=A$, let $last(w)$ be the last letter of $w$
to appear in a right-to-left scan of $w$.
It will be enough to find, for each $a\in A$, a sentence
$\phi_a$ of $\fobet$ defining $\{w\in h^{-1}(m):last(w)=a\}$,
since we then get $\psi_A$ as
\[\exists y(a(y)\wedge\forall x(x>y\rightarrow\neg a(x))\wedge\bigwedge_{b\neq a}\exists x(x>y\wedge b(x)))\wedge\phi_a.\]

A word $w$ with $\alpha(w)=A$ and $last(w)=a$ has a unique factorization
$w=uv$, where $\alpha(u)=A\backslash\{a\}$, and $v$ is an $a$-word.
We consider all factorizations $m=m_1m_2$ in $M$.
By the inductive hypothesis, there is a sentence $\mu$ of $\fobet$
defining the set of all words over $A\backslash\{a\}$ that map to $m_1$.
Suppose that we are able to find a sentence $\nu$ defining the set of all
$a$-words mapping to $m_2$.
We can then use a simple relativizing trick to obtain a sentence defining
all concatenations $uv$ such that $u\models\mu$ and $v\models\nu$.
One simply modifies each quantified subformula $\exists x\zeta$ of $\mu$ and
$\nu$, starting from the outside, changing them to
\[\exists x(\neg\exists y(y\leq x\wedge a(y)))~\mbox{and}~
\exists x(\exists y(y\leq x\wedge a(y))).\]
The conjunction of the two modified sentences now says that $\mu$ holds
in the factor preceding the first occurrence of $a$, and
$\nu$ holds in the factor that begins at the first occurrence of $a$.
Take the disjunction of these conjunctions over all factorizations $m_1m_2$
of $m$ to obtain $\phi_a$.

It remains to show how to construct a sentence that defines the set of
$a$-words that map to a given element $s$ of $M$.
Let $w\in A^*$ be an $a$-word. Let $\sigma$ be the final factorization scheme
in our sequence, so that
\[\sigma(w)=(w_1,\ldots,w_k),~~~ \sigma_h(w)=m_1\cdots m_k\in M^+.\]
In fact, each $w_i$ can be mapped to the subalphabet
\[N=\{h(v)\in M:\alpha(v)=A,v\in aA^*\},\]
so we can restrict to this subalphabet $N$ of $M$.

The map $n\mapsto n$ extends to a homomorphism from $N^+$
into the subsemigroup $S$ of $M$ generated by the elements of $N$.
Since the generators of $S$ are images of words $v$ with $\alpha(v)=A$,
we have $eSe \subseteq eM_{e}e$, which is in \da for every idempotent
$e\in E(S)$ by definition of $\meda$.
Thus the set of words over $N$ multiplying to $s\in S$ is defined
by a sentence $\psi$ over $N$ in $\fosucc$~\cite{TW}.
We can take the conjunction of this with a sentence that says
every letter belongs to the alphabet $N$, and thus obtain a sentence
$\psi'$ over $M$, also in $\fotwosuc$, defining this same set of words.
Thus by the Simulation Lemma~\ref{lem.simulations}, there is a sentence
$\phi$ in $\fobet$ that defines the set of $a$-words that map to $s$.
This completes the proof.


\section{Characterization of \texorpdfstring{$\fobetfac$}{FO2[<,BETFAC]}}\label{sec:cslnext}

The class of languages definable in the logic $\fobetfac$ corresponds to
a variety of finite semigroups rather than monoids.
We employ an operation
\[{\bf V}\mapsto {\bf V}*{\bf D}\]
that takes varieties of finite monoids to varieties of finite semigroups.  Normally
this is defined in terms of the semidirect product. However we have no need to introduce
this product here.  Instead we will rely on a characterization of this operation from~\cite{Str-v*d} and stated as Theorem~\ref{thm.v*d} below.  For purposes of this paper, this can
be taken as the definition of ${\bf V}*{\bf D}$.

Incidentally, the same operation underlies the characterization of $\fotwosuc$ that we have already
cited several times:  In fact, what Th\'erien and Wilke showed in~\cite{TW} is that a language $L$ is definable
in this logic if and only if $S(L)\in{\bf DA}*{\bf D}$.  They then relied on a difficult result of Almeida~\cite{Almeida} to
establish the effective characterization we used here. It is also the operation that connects the levels of the quantifier
alternation hierarchy in $FO[<]$ with the corresponding levels in $FO[<,\suc]$~\cite{Str-v*d}.

Here is our main result.

\begin{thm}[$\fobetfac$ characterizes $\medad$]\label{thm.betfac}
Let $L\subseteq A^+$.  $L$ is definable in $\fobetfac$ if and only if
$S(L)\in \medad$. Moreover, there is an effective procedure for determining
if $S(L)\in\medad$.
\end{thm}

We can use this theorem to establish some of the inclusions and non-inclusions depicted in Figure 1.
Since $\meda$ contains $\Delta_3[<]$ in the quantifier alternation hierarchy~\cite{Weil}, $\medad$ contains $\Delta_3[<,\suc]$, which includes
the language $BB_2 = {(a{(ab)}^*b)}^+$, which we showed in Corollary~\ref{cor.strictinclusion}
below is not in $\meda$. On the other hand it does not contain
$BB_3 ={(a{(a{(ab)}^*b)}^*b)}^+$.
Consider the language $U_3$ which is a sublanguage of $A^*c{(a+b)}^*cA^*$
such that between the marked $c$'s, the factor $bb$ does not occur before
the factor $aa$. This is in $\medad$ since it is defined by
the $\Pi_2[<,\suc]$ sentence
\[\begin{array}{ll}
\forall x\forall y\forall z\forall z'( & c(x) \land c(y) \land x < z < z' < y
\land \suc(z,z') \land b(z) \land b(z') \\
 & \limplies \exists w\exists w'
(x < w < w' < z \land \suc(w,w') \land a(w) \land a(w'))).
\end{array}\]

We now proceed to the proof.  We first state the characterization of ${\bf V}*{\bf D}$ in terms of ${\bf V}$ from~\cite{Str-v*d}.
This can be stated in several different ways, but all depend on some scheme for treating words of length $k$ over $A$ as individual letters. Here is a standard version.  Let $k>0$.  Let $A$ be a finite alphabet, and let $B=A^k$.  We treat $B$ as a finite alphabet itself---to distinguish the \textit{word} $w\in A^*$ of length $k$ from the same object considered as a \textit{letter} of $B$, we write $\expand{w}$ in the latter case.  We will  define, for a word $w\in A^+$ with $|w|\geq k-1$, a new word $w'\in B^*$, where $w'$ is simply the sequence of length-$k$ factors of $w$.  So, for example, with $A=\{a,b\}$ and $k=3$, if $w=aa$, then $w'=1\in B$, while if $w=ababba$, then
\[w'=\expand{aba}\expand{bab}\expand{abb}\expand{bba}.\]
To make sure that the lengths match up,
we supplement $A$ with a new symbol $*$ and define $B'$ as ${(A\cup\{*\})}^k$,
and $w''$ as the sequence of length-$k$ factors of $*^{k-1}w$.
For example, with this new definition, if $k=3$ and $w=ababba$, then
\[w''=\expand{**a}\expand{*ab}\expand{aba}\expand{bab}\expand{abb}\expand{bba}.\]

\begin{thm}[characterization of ${\bf V}*{\bf D}$~\cite{Str-v*d}]\label{thm.v*d}
Let $h:A^+\to S$ be a homomorphism onto a finite semigroup.
$S\in {\bf V}*{\bf D}$ if and only if there exist:
an integer $k>1$, and a homomorphism $h':B^*\to M\in{\bf V}$,
where $B=A^k$, such that whenever $v,w\in A^+$ are words that have
the same prefix of length $k-1$, and the same  suffix of length $k-1$,
and $v',w'$ are the sequence of $k$-length factors of $v,w$ respectively,
with $h'(v')=h'(w')$, then $h(v)=h(w)$.
\end{thm}

In brief, you can determine $h(w)$ by looking at the prefix and suffix of $w$
of length $k-1$, and checking the value of $w'$ under a homomorphism
$h'$ into an  element of ${\bf V}$.
Note that the statement is false if {\bf V} is the trivial variety
(and only in this case).

We will also need the following two propositions.

\begin{prop}[Delay]\label{simulate_betfac}
Let  $\phi$ be a sentence of $\fobetfac$.
Then there exist  $k>1$ and a sentence $\phi'$ of $\fobet$
interpreted over ${(A\cup\{*\})}^k$, with this property:
if $w\in A^+$ with $|w|\geq k-1$, then
$w\models\phi$ if and only if $w''\models\phi'$.
\end{prop}

\begin{prop}[Expansion]\label{simulate_bet}
Let $\phi'$ be a sentence of $\fobet$ interpreted over ${(A\cup\{*\})}^k$,
where $k>1$. Then there is a sentence $\phi$ of $\fobetfac$ with this property:
if $w\in A^+$ with $|w|\geq k-1$, then
$w\models\phi$ if and only if $w''\models\phi'$.
\end{prop}

Assuming these for now, we proceed to the proof of our characterization theorem.

\begin{proof}[Proof of Characterization Theorem~\ref{thm.betfac}]
Let $L\subseteq A^+$,  and suppose that $L$ is definable by a sentence $\phi$ of  $\fobetfac$.  Let $k>1$ and $\phi'$ in $\fobet$ be as given by Proposition~\ref{simulate_betfac}. Let $L'\subseteq {({(A\cup\{*\})}^k)}^*$ be the language defined by $\phi'$.  We will show that $S(L)\in \medad$.

Let $h:A^+\to S(L)$ be the syntactic morphism of $L$.
Let $h'$ be the syntactic morphism of $L'$ and
let $h''$ be the  restriction of $h'$ to elements of ${(A^k)}^*$.
Since $\phi'$ is a sentence of $\fobet$, the syntactic monoid of $L'$,
and hence the image of $h''$, belongs to $\meda$.
It is therefore enough, in view of Theorem~\ref{thm.v*d}, to suppose that
$v,w\in A^+$ have the same prefix of length $k-1$ and the same suffix of
length $k-1$, and that $h''(v')=h''(w')$, and then conclude that $h(v)=h(w)$.
To show $h(v)=h(w)$ we must show that for any $x,y\in A^*$,
$xvy\in L$ if and only if $xwy\in L$.
Given the symmetric nature of the statement,
it is enough to show $xvy\in L$ implies $xwy\in L$.
So let $xvy\in L$.  Then $xvy\models\phi$, so $(xvy)''\models\phi'$.
We take apart $(xvy)''$: Suppose
$x=a_1\cdots a_r, v=b_1\cdots b_s, y=c_1\cdots c_t$.

The leftmost $r+k-1$ letters of $(xvy)''$ are
\[\expand{*^{k-1}a_1}\expand{*^{k-2}a_1a_2}\cdots\expand{a_r b_1\cdots b_{k-1}}.\]
The rightmost $t$ letters of $(xvy)''$ are
\[\expand{b_{s-k+2}\cdots b_s c_1}\expand{b_{s-k+3}\cdots b_s c_1 c_2}\cdots\expand{c_{t-k+1}\cdots c_t}.\]
(The exact form of the last factor will be different if $t<k-1$.)
In between these two factors, we have the $s-k+1$ letters of $v'$.  Thus
$h'((xvy)'')=m_1 h''(v')m_2$,
where $m_1,m_2$ depend only on $x,y$ and the prefix and suffix of $v$ of length at most $k-1$.  It follows that we likewise have
$h'((xwy)'')=m_1 h''(w')m_2$,
with the same $m_1,m_2$.  Since $h''(v')=h''(w')$, we conclude
$h'((xvy)'')=h'((xwy)'')$, so $(xwy)''\models\phi'$.
Thus $xwy\models\phi$, and so $xwy\in L$.
This concludes the proof that $S(L)\in\medad$.

Conversely, suppose $L\subseteq A^+$ and that $S(L)\in \medad$.
Let $h:A^+\to S(L)$ be the syntactic morphism of $L$.  Let
$h':{(A^k)}^*\to M\in\meda$
be the homomorphism given by Theorem~\ref{thm.v*d}.
We extend $h'$ to ${({(A\cup\{*\})}^k)}^*$ by defining
$h'(b)=1$ for any $b$ that contains the new symbol $*$.
Then for each $m\in M$, we have a sentence $\phi_m'$ of $\fobet$
interpreted over ${({(A\cup\{*\})}^k)}^*$ defining ${(h')}^{-1}(m)$.
Let $\phi_m$ be the sentence over $\fobetfac$ given by
Proposition~\ref{simulate_bet}.
For each $x\in A^{k-1}$, let $\text{pref}_x$ be a sentence defining
the set of strings over $A$ whose prefix of length $k-1$ is $x$,
and similarly define $\text{suff}_x$.
Observe that both of these sentences can be chosen to be in  $\fobetfac$.
In fact, these properties are definable in $\fobet$ over $A$.
It follows that the set of words in $A^+$ of length at least $k-1$ mapping to
a given value $s$ of $S(L)$ is given by a disjunction of
finitely many sentences of the form
\[\text{pref}_x\wedge\text{suff}_y\wedge\phi'_m.\]
We thus get the complete preimage $h^{-1}(s)$ by taking the disjunction
with a sentence that says the word lies in a particular finite set.
So $L$ itself is definable in $\fobetfac$.

The claim about an effective procedure for testing  definability follows directly from another result in~\cite{Str-v*d}:  If  the  monoid variety {\bf V} has a decidable membership problem, and if {\bf V} contains the
syntactic monoid of the language ${(ab)}^*$, then the semigroup variety ${\bf V}*{\bf D}$ also has a decidable membership problem.
\end{proof}

We now provide proofs of the two propositions.

\begin{proof}[Proof of Delay Proposition~\ref{simulate_betfac}]
Let $\phi$ be a sentence of $\fobetfac$.
We choose for $k$ the length of the longest word $v\in A^+$ such that  $\diam{v}(x,y)$ occurs as a subformula of $\phi$. Because $w$ and $w''$ have the same length, the sets of positions in the two words are the same.  We thus keep occurrences of quantifiers and $x<y$ in $\phi$ exactly as they are, and consider how to translate the other atomic formulas in $\phi$.  First let's see how we translate $a(x)$.  Consider, for example,

\[w=bbacbab\]
with $k=4$.  We have
\[w''=\expand{***b}\expand{**bb}\expand{*bba}\expand{bbac}\expand{bacb}\expand{acba}\expand{cbab}.\]
The positions in $w$ that contain the letter $a$ correspond to the positions in $w''$ contain a word of length $k$ over $A\cup\{*\}$ that ends in $a$.  We thus translate $a(x)$ by
\[\bigvee \expand{ua}(x),\]
where the disjunction is over all $u\in {(A\cup\{*\})}^{k-1}$.

How do we translate $\diam{v}(x,y)$?  We might be tempted at first to try something similar, writing
\[\bigvee\expand{uv}(x,y),\]
where $u$ ranges over all $u$ of length $k-|v|$.
The problem is that $\expand{uv}(x,y)$ can hold without
$\diam{v}(x,y)$ being true.
Consider again the example above, with the variable $x$ pointing to
position 3 (treating 1 as the leftmost position of the word),
and $y$ to position 7.  With these choices of $x$ and $y$,
we have $\expand{bbac}(x,y)$ satisfied by $w''$,
however $\diam{ac}(x,y)$ is not satisfied by $w'$.
Observe that saying there is an occurrence of the factor $ac$ in $w$ between
$x$ and $y$  is equivalent to saying that the number of occurrences of letters
of the form $\expand{uac}$, with $|u|=2$,  in $w''$ between $x$ and $y$ is
greater than the number of occurrences of letters of this form in positions
$x+1$ and $x+2$. Thus the condition $\diam{ac}(x,y)$ is equivalent to the conjunction of conditions like
\begin{quote}
there is an occurrence of a letter of the form $\expand{uac}$ at position $x+1$, but not at position $x+2$, and at least two occurrences of letters of the form $\expand{uac}$ between $x$ and $y$
\end{quote}
We can write that there is, for example, an occurrence of $\expand{bbac}$ at $x+1$ but no occurrence of $\expand{abac}$ at position $x+2$ as
\[\exists y(y=x+1\wedge\expand{bbac}(y)\wedge\exists x(x=y+1\wedge\neg\expand{abac}(x))).\]
We write that there is at least one occurrence of $\expand{cbac}$ and at least one occurrence of $\expand{baac}$ between $x$ and $y$.
Observe that only a finite number of such conditions will be involved, since we never have to deal with thresholds larger than 2
(or $k-|v|$ in the general case).
In this manner we translate $\diam{v}(x,y)$ into a (rather complicated) formula
of two-variable logic with thresholds. As we showed in
Theorem~\ref{thm.invequalsthr}, this is definable in $\fobet$.
\end{proof}

\begin{proof}[Proof of Expansion Proposition~\ref{simulate_bet}]
We are given a sentence $\phi'$ of $\fobet$ interpreted over ${(A\cup\{*\})}^k$, for some $k>1$.  As before, there is a one-to-one correspondence between the positions of $w''$ and those of $w$.  As above, the only issue is showing how to translate $\expand{v}(x)$ and $\expand{v}(x,y)$.

The translation of $\expand{v}(x)$ depends on whether the word $v$ contains any occurrences of the symbol $*$.  For example, with $k=3$ again, and $v=bac$, the translation is
\[c(x)\wedge\exists y(x=y+1\wedge a(y)\wedge\exists x(y=x+1\wedge b(x))),\]
while with $v=*ac$ we get
\[c(x)\wedge\exists y(x=y+1\wedge a(y)\wedge\neg\exists x(y=x+1)).\]
To translate the $\expand{v}(x,y)$, we use the observation above, but here it is much simpler to apply:
$\expand{v}(x,y)$ holds in $w''$ if and only if either $\diam{v}(x,y)$ holds in $w$,
or $\expand{v}$ is the letter occurring at one of the positions $x+1,\ldots x+|v|-1$,
and we have already seen how to translate this latter condition.
\end{proof}

Finally, we note that the essentially the same argument gives us a proof of the equivalence of $\fotwothrfac$ and $\fotwobetfac$  (Theorem~\ref{thm.invequalsthrfac}): The two propositions providing the translation between formulas of $\fotwobet$ and $\fotwobetfac$ can be adapted to provide similar translations between $\fotwothr$ and $\fotwothrfac$. Since we know $\fotwobet=\fotwothr$ (Theorem~\ref{thm.invequalsthrfac}), we have that $L$ is definable in $\fotwothr$ if and only if $M(L)\in{\bf M}_e{\bf DA}$, and the methods of this section then imply that a language is definable in $\fotwothrfac$ if and only if its syntactic semigroup is in $M(L)\in{\bf M}_e{\bf DA}*{\bf D}$. Thus $\fotwothrfac$ and $\fotwobetfac$ define the same family of languages.


\section*{Acknowledgments}
Much of this research was carried out while the various authors were guests of
Boston College, the Tata Institute of Fundamental Research in Mumbai, the Chennai Mathematical Institute and the Institute of Mathematical Sciences in Chennai, and the University of Montreal, and participated in Dagstuhl Seminar 15401 on `Circuits, Logic and Games' in September, 2015, and in subsequent visits
among our institutions.
We would like to thank our institutions, Boston College, the Institute of Mathematical Sciences
and the Tata Institute of Fundamental Research, for hosting our collaborative visits.

\bibliographystyle{alpha}
\bibliography{algdecbet}
\end{document}